\documentclass[submission]{eptcs}
\pdfoutput=1

\usepackage[utf8]{inputenc}

\usepackage{tikz}
\usepackage{tikzit}
\usepackage{mathdots}
\usepackage{yhmath}
\usepackage{cancel}
\usepackage{color}

\usepackage{array}
\usepackage{multirow}
\usepackage{tabularx}
\usepackage{booktabs}
\usepackage{makecell}

\usepackage{hyperref}
\usepackage{todonotes}

\usepackage{extarrows}
\usepackage{braket}
\usepackage{enumitem}

\usepackage{amsmath,amsthm}
\usepackage{stmaryrd}
\usepackage{siunitx}
\usepackage{amssymb}
\usepackage{gensymb}
\usepackage{mathtools}
\usepackage{skak}

\setlist{topsep=0pt,itemsep=0.2pt,partopsep=2pt, parsep=2pt}
\usetikzlibrary{fadings}
\usetikzlibrary{patterns}
\usetikzlibrary{shadows.blur}
\usetikzlibrary{shapes}

\DeclareRobustCommand{\harp}[1]{\mathpalette\harpoonvec{#1}}
\newcommand{\harpvecsign}{\scriptscriptstyle\rightharpoonup}
\newcommand{\harpoonvec}[2]{%
  \ifx\displaystyle#1\doalign{$\harpvecsign$}{#1#2}\fi
  \ifx\textstyle#1\doalign{$\harpvecsign$}{#1#2}\fi
  \ifx\scriptstyle#1\doalign{\scalebox{.6}[.9]{$\harpvecsign$}}{#1#2}\fi
  \ifx\scriptscriptstyle#1\doalign{\scalebox{.5}[.8]{$\harpvecsign$}}{#1#2}\fi
}
\newcommand{\doalign}[2]{%
 {\vbox{\offinterlineskip\ialign{\hfil##\hfil\cr#1\cr$#2$\cr}}}%
}

\tikzstyle{doubled}=[line width=1.5pt] 

\tikzstyle{dot}=[inner sep=0mm,minimum width=2mm,minimum height=2mm,draw,shape=circle]  
\tikzstyle{ddot}=[inner sep=0mm, doubled, minimum width=2.5mm,minimum height=2.5mm,draw,shape=circle]

\tikzstyle{pdot}=[inner sep=0mm, doubled, minimum width=2.5mm,minimum height=2.5mm,shape=circle]
\tikzstyle{phase dimensions}=[minimum size=6mm,font=\footnotesize,inner sep=0.2mm,outer sep=-2mm]

\tikzstyle{phase dot}=[pdot,phase dimensions]
\tikzstyle{wphase dot}=[dot, phase dimensions]

\tikzstyle{hadamard}=[fill=white,draw,inner sep=0.6mm,font=\footnotesize,minimum height=6mm,minimum width=8mm]

\tikzstyle{anti} = [fill=white,draw,inner sep=0.6mm,font=\footnotesize,minimum height=3mm,minimum width=3mm]

\tikzstyle{triang}=[regular polygon,regular polygon sides=3,draw,scale=0.75,inner sep=-0.75pt,minimum width=9mm,fill=white,regular polygon rotate=180]
\tikzstyle{triang_lesssep}=[regular polygon,regular polygon sides=3,draw,scale=0.75,inner sep=-4pt,minimum width=9mm,fill=white,regular polygon rotate=180, text depth=4mm]
\tikzstyle{triangdag}=[regular polygon,regular polygon sides=3,draw,scale=0.75,inner sep=-0.5pt,minimum width=9mm,fill=white]

\makeatletter
\newcommand{\boxshape}[3]{%
\pgfdeclareshape{#1}{
\inheritsavedanchors[from=rectangle] 
\inheritanchorborder[from=rectangle]
\inheritanchor[from=rectangle]{center}
\inheritanchor[from=rectangle]{north}
\inheritanchor[from=rectangle]{south}
\inheritanchor[from=rectangle]{west}
\inheritanchor[from=rectangle]{east}
\backgroundpath{
\southwest \pgf@xa=\pgf@x \pgf@ya=\pgf@y
\northeast \pgf@xb=\pgf@x \pgf@yb=\pgf@y

\@tempdima=#2
\@tempdimb=#3

\pgfpathmoveto{\pgfpoint{\pgf@xa - 5pt + \@tempdima}{\pgf@ya}}
\pgfpathlineto{\pgfpoint{\pgf@xa - 5pt - \@tempdima}{\pgf@yb}}
\pgfpathlineto{\pgfpoint{\pgf@xb + 5pt + \@tempdimb}{\pgf@yb}}
\pgfpathlineto{\pgfpoint{\pgf@xb + 5pt - \@tempdimb}{\pgf@ya}}
\pgfpathlineto{\pgfpoint{\pgf@xa - 5pt + \@tempdima}{\pgf@ya}}
\pgfpathclose
}
}}

\boxshape{NEbox}{0pt}{5pt}
\boxshape{SEbox}{0pt}{-5pt}
\boxshape{NWbox}{5pt}{0pt}
\boxshape{SWbox}{-5pt}{0pt}
\boxshape{EBox}{-3pt}{3pt}
\boxshape{WBox}{3pt}{-3pt}
\makeatother

\tikzstyle{map}=[draw,shape=NEbox,inner sep=2pt,minimum height=6mm,fill=white]
\tikzstyle{mapdag}=[draw,shape=SEbox,inner sep=2pt,minimum height=6mm,fill=white]
\tikzstyle{maptrans}=[draw,shape=SWbox,inner sep=2pt,minimum height=6mm,fill=white]
\tikzstyle{mapconj}=[draw,shape=NWbox,inner sep=2pt,minimum height=6mm,fill=white]

\tikzstyle{dmap}=[draw,doubled,shape=NEbox,inner sep=2pt,minimum height=6mm,fill=white]
\tikzstyle{dmapdag}=[draw,doubled,shape=SEbox,inner sep=2pt,minimum height=6mm,fill=white]
\tikzstyle{dmaptrans}=[draw,doubled,shape=SWbox,inner sep=2pt,minimum height=6mm,fill=white]
\tikzstyle{dmapconj}=[draw,doubled,shape=NWbox,inner sep=2pt,minimum height=6mm,fill=white]

\makeatletter
\pgfdeclareshape{cornerpoint}{
\inheritsavedanchors[from=rectangle] 
\inheritanchorborder[from=rectangle]
\inheritanchor[from=rectangle]{center}
\inheritanchor[from=rectangle]{north}
\inheritanchor[from=rectangle]{south}
\inheritanchor[from=rectangle]{west}
\inheritanchor[from=rectangle]{east}
\backgroundpath{
\southwest \pgf@xa=\pgf@x \pgf@ya=\pgf@y
\northeast \pgf@xb=\pgf@x \pgf@yb=\pgf@y

\pgfmathsetmacro{\pgf@shorten@left}{\pgfkeysvalueof{/tikz/shorten left}}
\pgfmathsetmacro{\pgf@shorten@right}{\pgfkeysvalueof{/tikz/shorten right}}

\pgfpathmoveto{\pgfpoint{0.5 * (\pgf@xa + \pgf@xb)}{\pgf@ya - 5pt}}
\pgfpathlineto{\pgfpoint{\pgf@xa - 8pt + \pgf@shorten@left}{\pgf@yb - 1.5 * \pgf@shorten@left}}
\pgfpathlineto{\pgfpoint{\pgf@xa - 8pt + \pgf@shorten@left}{\pgf@yb}}
\pgfpathlineto{\pgfpoint{\pgf@xb + 8pt - \pgf@shorten@right}{\pgf@yb}}
\pgfpathlineto{\pgfpoint{\pgf@xb + 8pt - \pgf@shorten@right}{\pgf@yb - 1.5 * \pgf@shorten@right}}
\pgfpathclose
}
}

\pgfdeclareshape{cornercopoint}{
\inheritsavedanchors[from=rectangle] 
\inheritanchorborder[from=rectangle]
\inheritanchor[from=rectangle]{center}
\inheritanchor[from=rectangle]{north}
\inheritanchor[from=rectangle]{south}
\inheritanchor[from=rectangle]{west}
\inheritanchor[from=rectangle]{east}
\backgroundpath{
\southwest \pgf@xa=\pgf@x \pgf@ya=\pgf@y
\northeast \pgf@xb=\pgf@x \pgf@yb=\pgf@y

\pgfmathsetmacro{\pgf@shorten@left}{\pgfkeysvalueof{/tikz/shorten left}}
\pgfmathsetmacro{\pgf@shorten@right}{\pgfkeysvalueof{/tikz/shorten right}}

\pgfpathmoveto{\pgfpoint{0.5 * (\pgf@xa + \pgf@xb)}{\pgf@yb + 5pt}}
\pgfpathlineto{\pgfpoint{\pgf@xa - 8pt + \pgf@shorten@left}{\pgf@ya + 1.5 * \pgf@shorten@left}}
\pgfpathlineto{\pgfpoint{\pgf@xa - 8pt + \pgf@shorten@left}{\pgf@ya}}
\pgfpathlineto{\pgfpoint{\pgf@xb + 8pt - \pgf@shorten@right}{\pgf@ya}}
\pgfpathlineto{\pgfpoint{\pgf@xb + 8pt - \pgf@shorten@right}{\pgf@ya + 1.5 * \pgf@shorten@right}}
\pgfpathclose
}
}

\makeatother

\pgfkeyssetvalue{/tikz/shorten left}{0pt}
\pgfkeyssetvalue{/tikz/shorten right}{0pt}

\tikzstyle{kpoint common}=[draw,fill=white,inner sep=1pt,minimum height=4mm]
\tikzstyle{kpoint}=[shape=cornerpoint,shorten left=5pt,kpoint common]
\tikzstyle{kpoint adjoint}=[shape=cornercopoint,shorten left=5pt,kpoint common]
\tikzstyle{kpoint conjugate}=[shape=cornerpoint,shorten right=5pt,kpoint common]
\tikzstyle{kpoint transpose}=[shape=cornercopoint,shorten right=5pt,kpoint common]

\tikzstyle{kpointdag}=[kpoint adjoint]
\tikzstyle{kpointadj}=[kpoint adjoint]
\tikzstyle{kpointconj}=[kpoint conjugate]
\tikzstyle{kpointtrans}=[kpoint transpose]

\tikzstyle{big kpoint}=[kpoint, minimum width=1.0 cm, minimum height=2mm, inner sep=4pt, text depth=1.5mm]

 \tikzstyle{upground}=[circuit ee IEC,thick,ground,rotate=90,scale=1.5]
 \tikzstyle{downground}=[circuit ee IEC,thick,ground,rotate=-90,scale=1.5]

\tikzstyle{discarding}=[fill=white, draw=black, shape=circle, style=upground]
\tikzstyle{smalldiscarding}=[fill=white, draw=black, style=upground, scale=0.5]
\tikzstyle{backdiscard}=[fill=white, draw=black, shape=circle, style=downground, scale=0.5]
\tikzstyle{smallbackdiscard}=[fill=white, draw=black, shape=circle, style=downground, scale=0.5]
\tikzstyle{state}=[fill=white, draw=black, style=triang, tikzit shape=rectangle]
\tikzstyle{kstate}=[fill=white, draw=black, style=kpoint, tikzit shape=rectangle]
\tikzstyle{kstateconj}=[fill=white, draw=black, style=kpoint conjugate, tikzit shape=rectangle]
\tikzstyle{kstateBIG}=[fill=white, draw=black, style=big kpoint, tikzit shape=rectangle]
\tikzstyle{effect}=[fill=white, draw=black, style=triangdag]
\tikzstyle{keffect}=[fill=white, draw=black, style=kpoint adjoint]
\tikzstyle{keffectconj}=[fill=white, draw=black, style=kpoint transpose]
\tikzstyle{morphdag}=[style=mapdag]
\tikzstyle{morph}=[style=hadamard]
\tikzstyle{WIDEmorph}=[style=hadamard, minimum width=14mm]
\tikzstyle{morphtrans}=[style=maptrans]
\tikzstyle{morphconj}=[style=mapconj]
\tikzstyle{CPMmorph}=[style=dmap]
\tikzstyle{CPMmorphconj}=[style=dmapconj]
\tikzstyle{CPMmorphdag}=[style=dmapdag]
\tikzstyle{CPMmorphtrans}=[style=dmaptrans]
\tikzstyle{CPMstate}=[fill=white, draw=black, style=triang, doubled]
\tikzstyle{CPMstateBIG}=[fill=white, draw=black, style={triang_lesssep}, doubled]
\tikzstyle{CPMkstate}=[fill=white, draw=black, style=kpoint, tikzit shape=rectangle, doubled]
\tikzstyle{CPMkstateconj}=[fill=white, draw=black, style=kpoint conjugate, tikzit shape=rectangle, doubled]
\tikzstyle{CPMkstateBIG}=[fill=white, draw=black, style=big kpoint, tikzit shape=rectangle, doubled]
\tikzstyle{CPMkeffect}=[fill=white, draw=black, style=kpoint adjoint, doubled]
\tikzstyle{CPMkeffectconj}=[fill=white, draw=black, style=kpoint transpose, doubled]
\tikzstyle{UHfB}=[fill=white, draw=black, style=triangdag, doubled, inner sep=-2pt]
\tikzstyle{leak}=[style=tinypoint, regular polygon rotate=-90]
\tikzstyle{leakfill}=[style=tinypoint, regular polygon rotate=-90, fill=black]
\tikzstyle{Z}=[style=dot, fill=green]
\tikzstyle{X}=[style=dot, fill=red]
\tikzstyle{black_dot}=[style=dot, fill=black]
\tikzstyle{white_dot}=[style=dot, fill=white]
\tikzstyle{qblack_dot}=[style=ddot, fill=black]
\tikzstyle{qwhite_dot}=[style=ddot, fill=white]
\tikzstyle{whitephase}=[style=wphase dot, fill=white]
\tikzstyle{qredphase}=[style=phase dot, fill=red]
\tikzstyle{qgreenphase}=[style=phase dot, fill=green]
\tikzstyle{had}=[style=hadamard, doubled]
\tikzstyle{box}=[style=hadamard]
\tikzstyle{classhad}=[style=hadamard]
\tikzstyle{antipode}=[style=anti]

\tikzstyle{dottededge}=[-, dotted]
\tikzstyle{double edge}=[-, style=doubled, draw=black, tikzit draw={rgb,255: red,191; green,0; blue,64}]
\tikzstyle{new edge style 0}=[<-]

\usetikzlibrary{calc, positioning, shapes.geometric}
\usetikzlibrary{
	arrows,
	shapes,
	decorations,
	intersections,
	backgrounds,
	positioning,
	circuits.ee.IEC
	}

\DeclareRobustCommand{\tensormuchi}{\!
\mathbin{\begin{tikzpicture}[scale=0.45, every node/.style={scale=0.6}]
	\begin{pgfonlayer}{nodelayer}
		\node [style=none] (0) at (0, 0.25) {};
		\node [style=none] (3) at (0, -0.25) {};
		\node [style=none] (4) at (0, 0) {$\mu \chi$};
	\end{pgfonlayer}
	\begin{pgfonlayer}{edgelayer}
		\draw [bend left=90, looseness=2.50] (0.center) to (3.center);
		\draw [bend right=90, looseness=2.50] (0.center) to (3.center);
	\end{pgfonlayer}
\end{tikzpicture}}  \!}

\DeclareRobustCommand{\tensormubarchi}{\!
\mathbin{\begin{tikzpicture}[scale=0.45, every node/.style={scale=0.6}]
	\begin{pgfonlayer}{nodelayer}
		\node [style=none] (0) at (0, 0.25) {};
		\node [style=none] (3) at (0, -0.25) {};
		\node [style=none] (4) at (0, 0) {$\overline{\mu} \chi$};
	\end{pgfonlayer}
	\begin{pgfonlayer}{edgelayer}
		\draw [bend left=90, looseness=2.50] (0.center) to (3.center);
		\draw [bend right=90, looseness=2.50] (0.center) to (3.center);
	\end{pgfonlayer}
\end{tikzpicture}}  \!}

\DeclareRobustCommand{\tensorchi}{ \!
\mathbin{ \begin{tikzpicture}[scale=0.45, every node/.style={scale=0.6}]
	\begin{pgfonlayer}{nodelayer}
		\node [style=none] (0) at (0, 0.25) {};
		\node [style=none] (3) at (0, -0.25) {};
		\node [style=none] (4) at (0, 0) {$\chi$};
	\end{pgfonlayer}
	\begin{pgfonlayer}{edgelayer}
		\draw [bend left=90, looseness=1.75] (0.center) to (3.center);
		\draw [bend right=90, looseness=1.75] (0.center) to (3.center);
	\end{pgfonlayer}
\end{tikzpicture}} \! }

\DeclareRobustCommand{\tensorzeta}{\! 
\mathbin{ \begin{tikzpicture}[scale=0.45, every node/.style={scale=0.6}]
	\begin{pgfonlayer}{nodelayer}
		\node [style=none] (0) at (0, 0.25) {};
		\node [style=none] (3) at (0, -0.25) {};
		\node [style=none] (4) at (0, 0) {$\zeta$};
	\end{pgfonlayer}
	\begin{pgfonlayer}{edgelayer}
		\draw [bend left=90, looseness=1.75] (0.center) to (3.center);
		\draw [bend right=90, looseness=1.75] (0.center) to (3.center);
	\end{pgfonlayer}
\end{tikzpicture}} \! }

\DeclareRobustCommand{\tensormu}{ \!
\mathbin{ \begin{tikzpicture}[scale=0.45, every node/.style={scale=0.6}]
	\begin{pgfonlayer}{nodelayer}
		\node [style=none] (0) at (0, 0.25) {};
		\node [style=none] (3) at (0, -0.25) {};
		\node [style=none] (4) at (0, 0) {$\mu$};
	\end{pgfonlayer}
	\begin{pgfonlayer}{edgelayer}
		\draw [bend left=90, looseness=1.75] (0.center) to (3.center);
		\draw [bend right=90, looseness=1.75] (0.center) to (3.center);
	\end{pgfonlayer}
\end{tikzpicture}} \! }

\DeclareRobustCommand{\tensorxi}{\! 
\mathbin{ \begin{tikzpicture}[scale=0.45, every node/.style={scale=0.6}]
	\begin{pgfonlayer}{nodelayer}
		\node [style=none] (0) at (0, 0.25) {};
		\node [style=none] (3) at (0, -0.25) {};
		\node [style=none] (4) at (0, 0) {$\xi$};
	\end{pgfonlayer}
	\begin{pgfonlayer}{edgelayer}
		\draw [bend left=90, looseness=1.75] (0.center) to (3.center);
		\draw [bend right=90, looseness=1.75] (0.center) to (3.center);
	\end{pgfonlayer}
\end{tikzpicture}} \! }

\newtheorem{thm}{Theorem}
\newtheorem{definition}{Definition}
\newtheorem{proposition}{Proposition}

\begin{document}

\title{Generalised tensors and traces} 

\author{
Pablo Arrighi \\ 
Université Paris-Saclay, Inria, CNRS, LMF, 91190 Gif-sur-Yvette, France \\ 
IXXI, Lyon, France \\
\texttt{pablo.arrighi@universite-paris-saclay.fr} \\
Amélia Durbec \\ 
Université Paris-Saclay, CNRS, LISN, 91190 Gif-sur-Yvette, France\\
\texttt{amelia.durbec@universite-paris-saclay.fr} \\
Matt Wilson\\
Department of Computer Science, University of Oxford, UK\\
HKU-Oxford Joint Laboratory for Quantum Information and Computation\\
\texttt{matthew.wilson@cs.ox.ac.uk}
}

\def\titlerunning{Generalised tensors and traces}
\def\authorrunning{Pablo Arrighi, Amélia Durbec, Matt Wilson}

\maketitle

\begin{abstract}
Tensors and traceouts are generalised, so that systems can be partitioned according to almost arbitrary logical predicates. One might have feared that the familiar interrelations between the notions of unitarity, complete positivity, trace-preservation, non-signalling causality, locality and localisability that are standard in quantum theory be jeopardized as the partitioning of systems becomes both logical and dynamical. Such interrelations in fact carry through, although a new notion, consistency, becomes instrumental. 
\end{abstract}

\section{Introduction}


On what grounds can one select `part' of a global system, and call it a `subsystem'? In textbook quantum theory one usually identifies subsystems as factors of tensor products. This is usually expressed w.r.t a basis, by identifying Cartesian factors within the basis, e.g. $\{\ket{0}, \ket{1}\}\times \{\ket{0}, \ket{1}\}$ within $\{\ket{00}, \ket{01}, \ket{10}, \ket{11}\}$. One could imagine, however, many other ways to select part of a system w.r.t a basis, based on `discriminating functions' or other logical predicates. For instance consider a collection of systems with names $u,v,w \dots$ and internal states given by some colors, e.g. black and white. The usual way to designate subsystems is simply by means of their names:  \[\includegraphics[scale = 0.2]{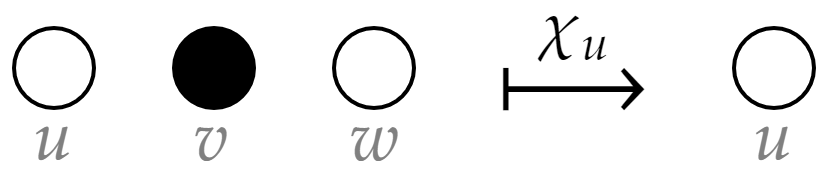}\] Instead, one could select those which are in a particular state, e.g. white: \[\includegraphics[scale = 0.2]{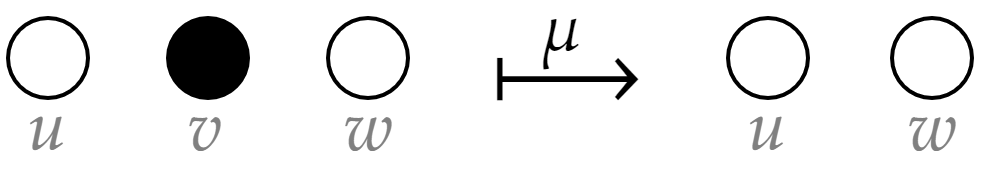}\] In a network one could ask for all of the neighbours of the node named $u$:  \[\includegraphics[scale = 0.2]{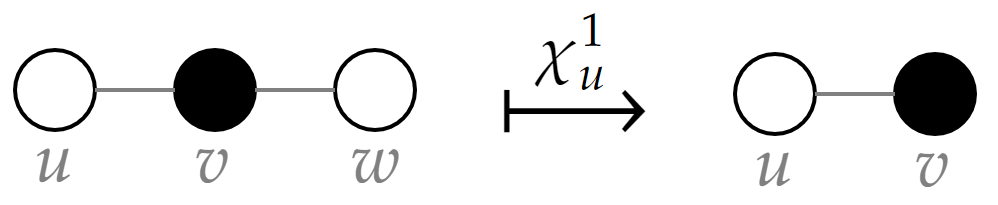}\] Given such a generalised notion of subsystem, one may also wonder whether there is a way to compute its reduced density matrix, by means of a \textit{generalised traceout}. One may also ask if there is a way of weaving subsystems back together, i.e. a \textit{generalised tensor product}.  In providing such a notion, we wish to disambiguate the fact that the same tensor symbol be used to mean two different things: \[\includegraphics[scale = 0.2]{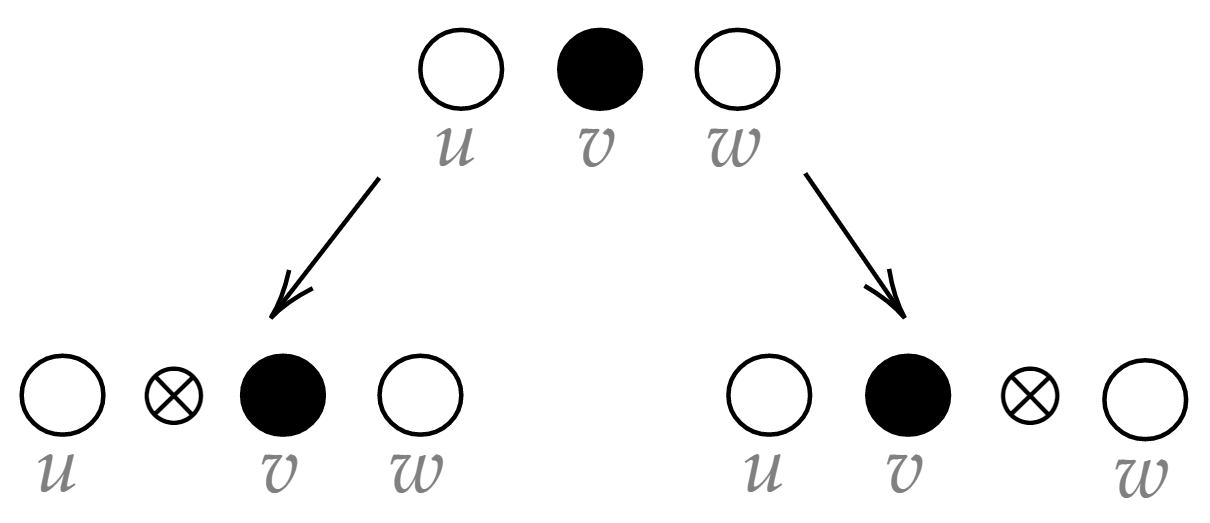}\] By writing: \[\includegraphics[scale = 0.2]{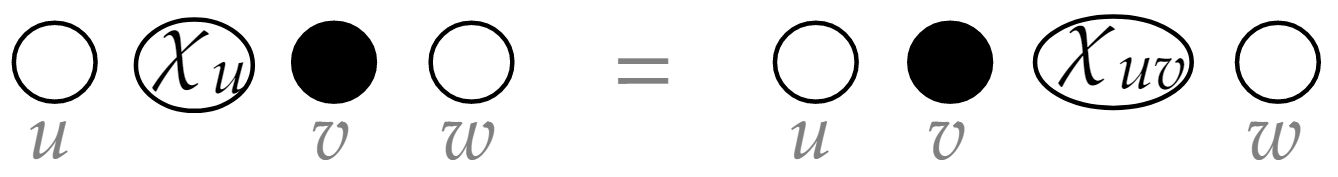}\] 
to state that the former decomposition is based upon $\chi_{u}$ and the latter upon $\chi_{uv}$. 
Or even \[\includegraphics[scale = 0.2]{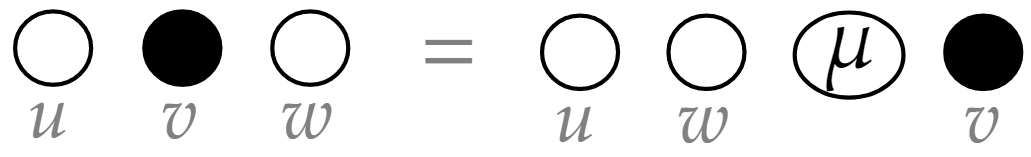}\]
i.e. weaving back together subsystems selected based on their states.\\


The need for a generalised notion of subsystem is in fact ubiquitous \cite{moliner_space, GogiosoChurch, chiribella_subsystems}. For instance it is well-known that a product state according to one partition into subsystems, may in fact provide useful entanglement according to another \cite{zanardi_observe, zanardi_virtual}. It is well-known also that there exists quantum evolution that have a strict causal structure, and yet cannot be decomposed as local gates upon standard subsystems in a way that follows this causal structure \cite{Lorenz_2021}. In the theory of quantum reference frames \cite{hamette_qrt, Spekkens_refs}, one aims at taking the perspective of a superposed observer: but what subsystems will such an observer see, with what corresponding reduced density matrices? Could such an observer be modelled as a token, and isolated as in the above $\mu$ example? Superposing entire quantum networks, the neighbourhood of a particular node becomes a state-dependent notion \cite{RovelliLQG,ArrighiQCGD}, which is something that the textbook notion of subsystem cannot handle. This is because the notion of neighbourhood needs be promoted to an operator (e.g. à la $\chi^1_\mu$). All of these issues will benefit from a more liberal notion of subsystem, albeit w.r.t. a basis. To make the notion of subsystem basis independent, algebraic quantum field theory tends to think of them as Von Neumann algebras instead \cite{BratteliRobinson}. Then, the Wedderburn-Artin theorem states that up to a unitary, the algebra is a direct sum of full matrix algebras tensored with the identity, i.e. still something that the textbook notion of subsystem cannot handle. Again this calls for a generalised notion of subsystems, able to discriminate both subspaces (e.g. à la $\mu$) and tensor factors (e.g. à la $\chi_u$), together with the corresponding notions of reduced density matrix and traceouts.


The generalization of the tensor product we propose follows from a decompositional (as opposed to compositional) approach. That is, we discriminate which nodes go left of a tensor $\tensorchi$, and which nodes go right, based on an almost arbitrary logical predicate $\chi$, so that for any configuration $G$, we have $\ket{G}=\ket{G_\chi}\tensorchi\ket{G_{\overline{\chi}}}$. So that the decomposition be unambiguous, we ask that whenever $L$ and $R$ are not `consistent', i.e. not of the form $L=G_\chi$ and $R=G_{\overline{\chi}}$ for some $G$, then $\ket{L}\tensorchi \ket{R}=0$. Surprisingly, any function $\chi$ mapping $G$ a network into $G_\chi \subseteq G$ a subnetwork and fulfilling the condition that $G_{\chi\chi}=G_{\chi}$ is fit for our purposes: we refer to these as `restrictions'. Restrictions then lead to a natural notion of partial trace $(\ket{G}\bra{H})_{|\chi}=\ket{G_\chi}\bra{H_\chi}\braket{H_{\overline{\chi}}|G_{\overline{\chi}}}$ which is itself a completely positive and trace-preserving operation. This generalised partial trace is robust, e.g. comprehension $\zeta \sqsubseteq \chi$ implies $(\rho_{|\chi})_{|\zeta}=\rho_\zeta$.

What tells us that these envisioned $\tensorchi$ are not just ill-behaved mathematical symbols, and that they do deserve to be called `tensors'? What is the purpose of `tensors', anyway? The reason why we use tensors, in quantum theory, to begin with, is order to decompose systems into subsystems endowed with robust notions of locality and causality. Thus, when seeking to generalise tensors, this must be our main criterion for success.

Indeed, recall that in textbook quantum theory an operator is local on $u$ if and only if $A = A' \otimes I_{\overline{u}}$ and equivalently if $Tr[A \rho] = Tr [A\rho']$ whenever $Tr_{\overline{u}}[\rho] = Tr_{\overline{u}}[\rho']$. Thus, locality can be expressed either `operationally' by means of the composition $\otimes$, or `observationally' by means of traceout $Tr_{\overline{u}}[\cdot]$ in the Heisenberg picture. Similarly causality can be expressed using traces as the satisfaction of no-signalling constraints in the Schrödinger picture, but also by means of a decomposition into tensors of local gates in the operational picture, or again as the preservation of the locality of observables in the Heisenberg picture. The above series of facts not only express a reasonable interaction between the principles of causality and locality, but may actually be taken as evidence that the tensors and traceouts of textbook quantum theory yield appropriate notions of composition and traceouts upon subsystems. This paper relies on the same evidence to show that one can generalise tensors, traceouts, locality, and causality to much broader notions of subsystems than picking subfactors of tensor products, handling situations where system population across the induced partition is quantised.

Intuitively, local operators alter only a part $\chi$ of the network, and ignore the rest. In this paper we say that an operator is $\chi$-local whenever $\bra{H}A\ket{G}=\bra{H_\chi}A\ket{G_\chi}\braket{H_{\overline{\chi}}|G_{\overline{\chi}}}$ and prove the equivalence with the requirement that $A=A\tensorchi I$ and $\textrm{Tr}(A\rho)=\textrm{Tr}(A\rho_{|\chi})$. This therefore establishes that the Schr\"odinger, operational and Heisenberg pictures ways of phrasing locality remain equivalent for this generalised notion of locality. 

Intuitively, causal operators act over the entire network, yet respecting that effects on region $\zeta$ be fully determined by causes in region $\chi$. In this paper we say that operator $U$ is causal when $(U\rho U^\dagger)_{|\zeta}=(U\rho_{|\chi} U^\dagger)_{|\zeta}$ and prove equivalence with asking for $A$ $\zeta$-local to imply $U^\dagger A U$ $\chi$-local. Causality refers to the physical principle according to which information propagates at a bounded speed, localisability refers to the principle that all must emerge constructively from underlying local mechanisms, that govern the interactions of closeby systems. The theorem shows that causality implies localisability. This therefore again establishes that the Schr\"odinger, operational and Heisenberg pictures ways of phrasing causality remain equivalent for this generalised notion of causality.


Another criterion for the success of these generalizations would be to exhibit potential applications. One application is that the generalised tensor product will be able to subsume both the old tensor and the direct sum for most purposes, at a moment where quantum calculi that account for both tensors and direct sums are trending \cite{Sheets,ManyWorlds}---as they are directly applicable for the study of indefinite causal order \cite{ConsistentICO} and local complete rewrite rules for quantum circuits \cite{QCircuitCompleteness}. Here we will also be able to formulate causality in some typically quantum contexts, for which we did not have the words before and yet were totally intuitive. In particular we will be able to give a precise meaning to the quantum operation taking oriented neighbourhood of radius $1$ around a node, in a scenario in which the neighbours are in a superposition. We will do this by modelling connectivity (left) between nodes with additional auxiliary nodes (right):
\begin{equation}\includegraphics[scale = 0.15]{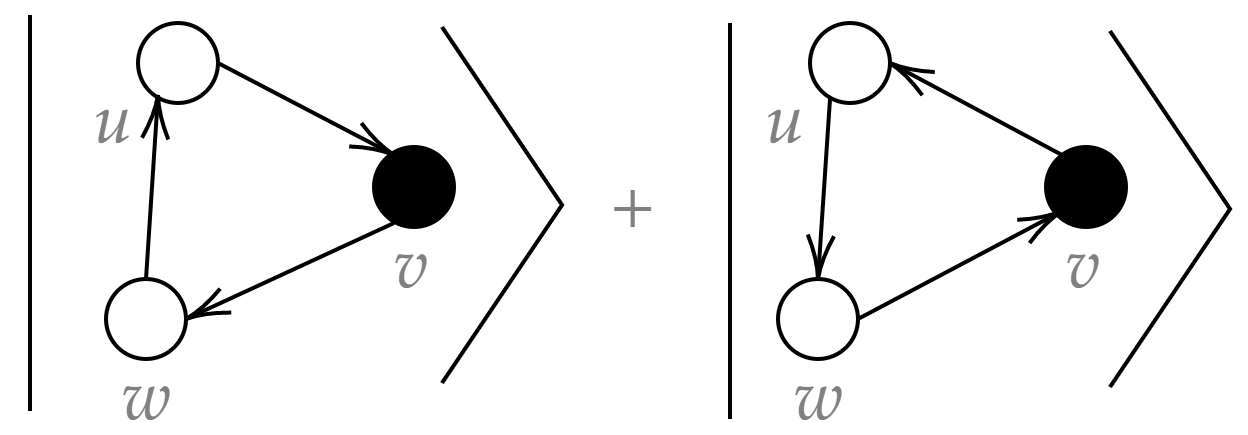} \quad \quad \quad \quad \includegraphics[scale = 0.15]{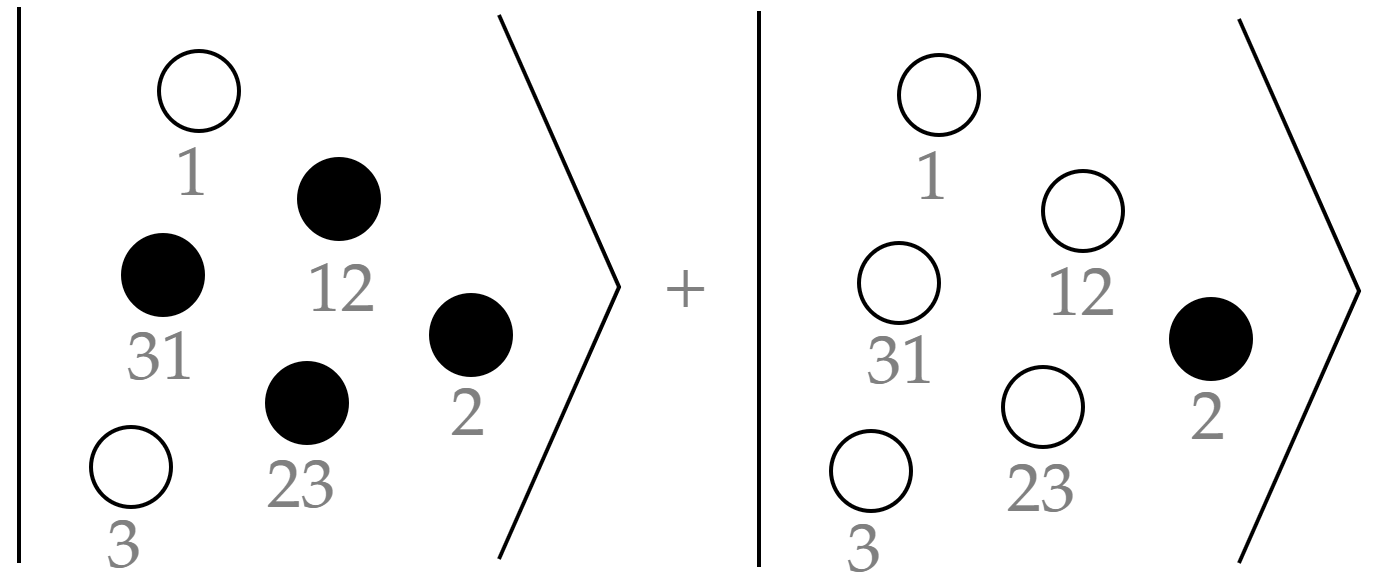}.\end{equation}

{\em Plan.} Sec. \ref{sec:tensors} defines configurations and their induced state space, generalised tensors and generalised traceouts. Sec. \ref{sec:locality} defines locality in its Schr\"odinger, operational and Heisenberg forms and proves their equivalences. Sec. \ref{sec:causality} defines causality in its various forms and proves their equivalences, including the fact that causal unitary operators can be implemented by a product of local ones. Sec. \ref{sec:conclusion} enumerates several perspective applications of the formalism. This conference paper was made as self-contained as possible, with most proofs in App. \ref{app:props} and that of the main Th. in App. \ref{app:th}. Yet some mathematical details are only available in the long version of this paper \cite{ArrighiQNT}, cited as [F] in the rest of the paper. 


\section{Generalised tensors and traces}\label{sec:tensors}

\paragraph*{State space} We take a `system' to mean both a `state' and a `name', whereas a `configuration' is a set of systems having disjoint names. We will typically represent a name with a letter $u,v \dots$ and represent its available states with colors: \[\includegraphics[scale = 0.2]{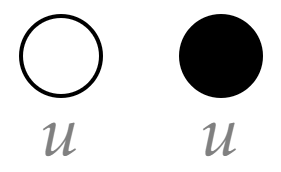}\] Formally we will pick names in a countable set of possible names $\mathcal{V}$, and states in some set $\Sigma$. In our working examples we refer to names of $\mathcal{V}$ with letters such as $u,v,w$ and use the state space $\Sigma := \{  black , \textrm{ }white \}$. We write $white.u$ to refer to a system named $u$ and having state $white$.
\begin{definition}[Systems]
For every pair of countable sets $\mathcal{V}$ and $\Sigma$ we call the elements of the set $S := \Sigma \times \mathcal{V}$ \textit{systems}. Each element $(\sigma , v)$ will be denoted $\sigma . v$ and for each $\sigma . v$ we call
\begin{itemize}
\item $ \sigma \in \Sigma $ the internal state of the system
\item $ v\in \mathcal{V}$ the vertex which supports the system
\end{itemize}
\end{definition}
Our formal definition of a configuration is given by the following:
\begin{definition}[Configurations]\label{def:graphs}
A configuration $ G$ is a finite set of systems such that
\begin{align}
\sigma .v,\ \sigma '.v'\in S \textrm{ and } v = v' \textrm{ implies } \sigma =\sigma' \label{eq:wellnamedness}
\end{align}
We denote $\varnothing$ the empty configuration.
We define the support of a configuration as $ V( G) :=\{v\ |\ \sigma .v\in G\}$. We denote by $ \mathcal{G}$ the set of all configurations.\\ We denote by $ \mathcal{H}$ the Hilbert space whose canonical basis is labelled by the elements of $ \mathcal{G}$.
\end{definition}

In other words, states are formal linear combinations of arbitrary configurations, of the form $\ket{\psi}=\sum_i \alpha_i \ket{G(i)}$. We will denote the canonical configuration basis using bra-ket notation in the following way: \[\includegraphics[scale = 0.2]{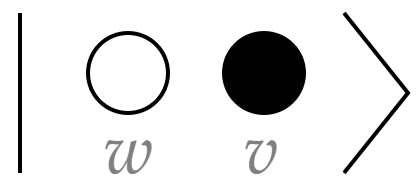}\] and take formal linear combinations over these. Notice that varying populations are allowed within the above defined Hilbert space ${\cal H}$ e.g. \[\includegraphics[scale = 0.2]{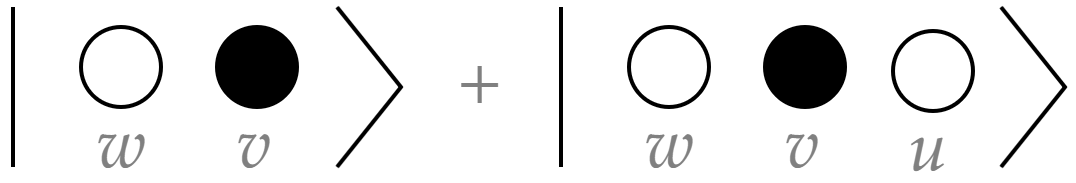}.\] 

\paragraph*{Restrictions, traceouts, tensors} We partition systems in a modular fashion, i.e. parametrised by predicates, referred to as 'restrictions'.

\begin{definition}[Restrictions, comprehension]\label{def:traceouts}

Consider a function
\begin{align*}
\chi &: \mathcal{G}\rightarrow \mathcal{G}\\
\chi &: G\mapsto G_{\chi } \subseteq G
\end{align*}
It is called a restriction if and only if $\chi \circ \chi = \chi$.\\
A restriction is pointwise if and only if $ G_{\chi } =\bigcup _{\sigma .v\ \in \ G}\{\sigma .v\}_{\chi }$.\\
A restriction is extensible if and only if $\ G_{\chi} \subseteq H\subseteq G \Rightarrow H_{\chi} = G_{\chi}$.\\
We use the notation $ G_{\overline{\chi }} :=G\setminus G_{\chi }$ even though $ \overline{\chi }$ is not necessarily a restriction.\\
We use the notation $ G_{\chi \zeta } :=( G_{\chi })_{\zeta }$ and thus $ \chi \zeta :=\zeta \circ \chi $.\\
We use the notation $ [ \chi ,\zeta ] =0$ to mean that $ \chi \zeta =\zeta \chi $. \\
We say that $\zeta$ is comprehended within $\chi$ and write $ \zeta \sqsubseteq \chi $ if and only if $ G_{\chi \zeta } =G_{\zeta }$ and 
$\braket{H_{\overline{\zeta }}|G_{\overline{\zeta }}} =\braket{H_{\chi \overline{\zeta }}|G_{\chi \overline{\zeta }}}\braket{H_{\overline{\chi }} | G_{\overline{\chi }} }.$
\end{definition}
Restrictions satisfy a variety of convenient properties, such as
 $\chi \chi=\chi$ and $\chi\overline{\chi}=\varnothing$ as one may expect, i.e. picking $\chi$ then picking $\chi$ again is the same as picking $\chi$ once. 
The simplest example of a restriction is $\chi_u$, as informally presented in the introduction. In more formal terms $\sigma.v \in G_{\chi_u} \iff \sigma.v \in G \textrm{ and } v = u$. Another simple pointwise example is the $\mu$ given by $\sigma.u \in G_{\mu} \iff \sigma.u \in G \textrm{ and } \sigma = white$.
 As a more complex example of a restriction, consider the set of names $\mathcal{V} := \{  1,2,3,12,31,23    \}$, and choose to interpret $black.xy$  as a directed edge from $x$ to $y$ and $white.xy$ as an edge from $y$ to $x$. We can then define restriction $\harp{\chi}_3^{r=1}$ which takes the incoming disk of radius $1$ around node $3$, for instance in the following configurations, on the left the node $white.31$ is interpreted as an edge from $1$ to $3$ meaning that it, along with $1$ is kept as part of $\harp{\chi}_3^{1}$. On the right, the node $black.23$ is interpreted by $\chi_3^{1}$ as an edge from $2$ to $3$ and so both $23$ and $2$ are kept by $\chi_{3}^{r=1}$:  \[\includegraphics[scale = 0.15]{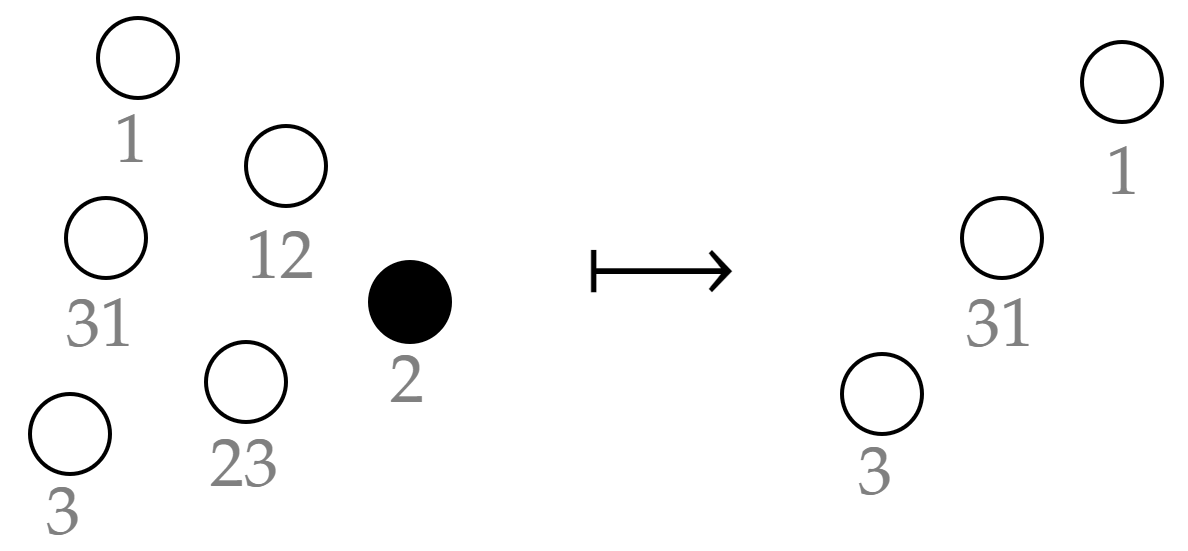} \quad \quad \quad \quad \includegraphics[scale = 0.15]{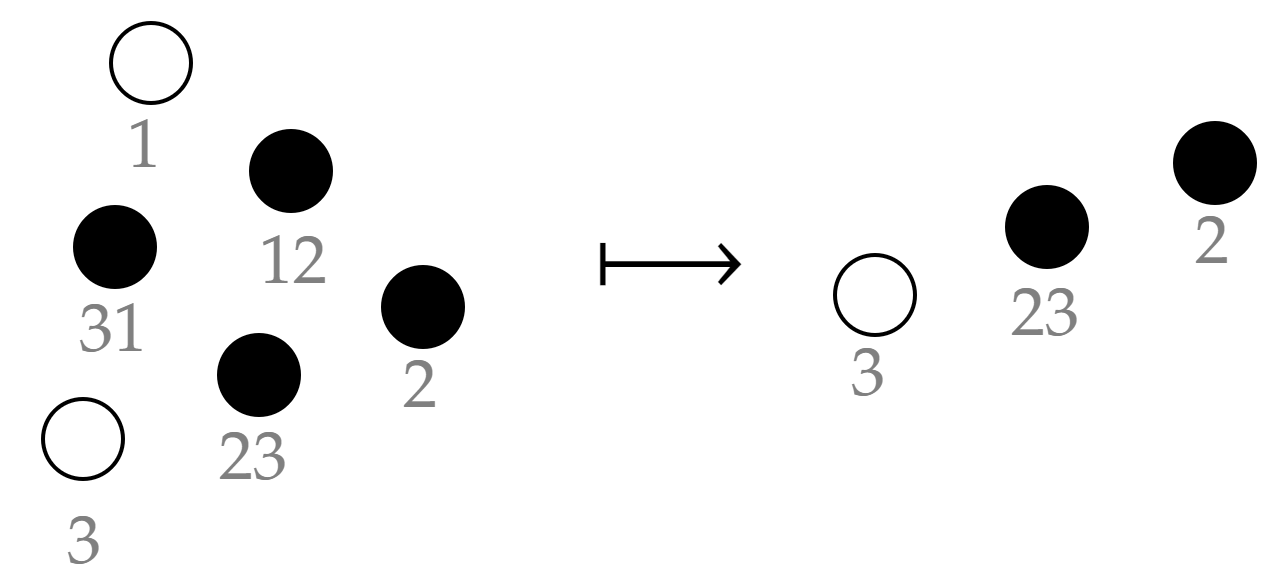}\] 



We now show how to trace away the environment of a system defined by restriction.
\begin{definition}[Traceout]
The $\chi$-trace $ \rho _{|\chi }$ of any trace class operator $\rho$ is defined by linear extension of the function $ (\ket{G}\bra{H})_{|\chi } := \ket{G_{\chi }} \bra{H_{\chi }} \braket{ H_{\overline{\chi }} | G_{\overline{\chi }}}$. We will denote by $ \rho _{|\emptyset }$ the usual, full trace $ \text{Tr}( \rho )$.
\end{definition}

Restrictions must not be confused with the partial traces $( .)_{|\chi }$ that they induce, as illustrated in the following example calculations for the pointwise restriction $\chi_{u}$: \[\includegraphics[width=0.75\textwidth]{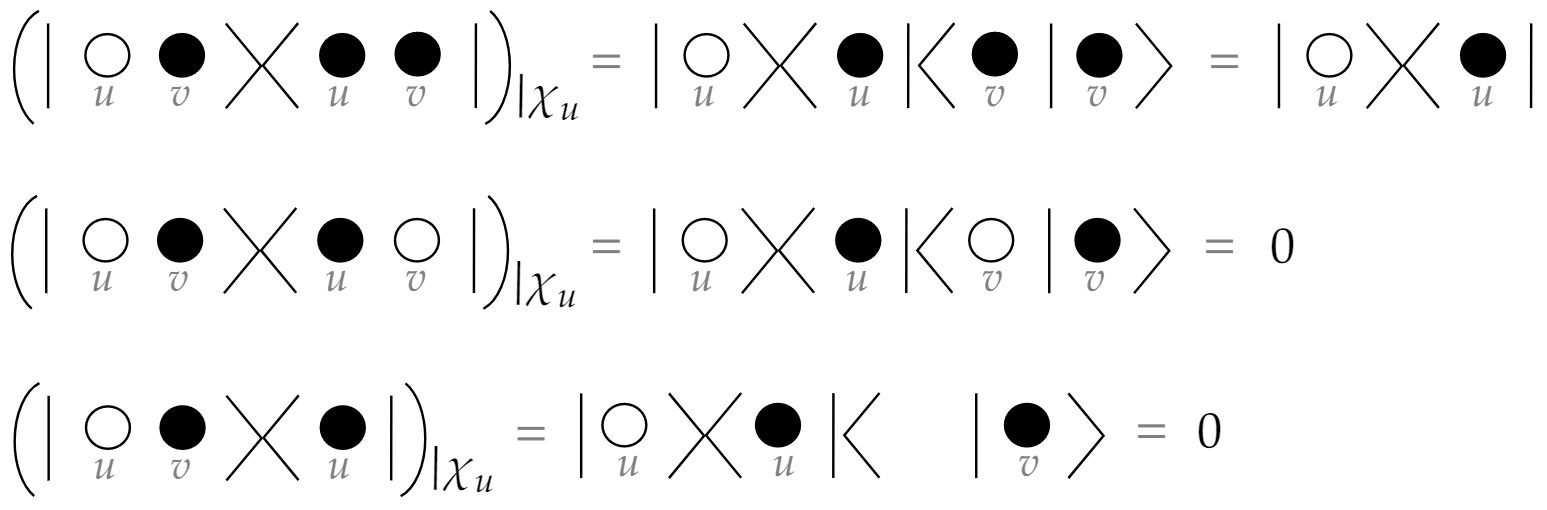}\] In the first line, the ket and bra coincide beyond $\chi_u$, allowing to leave behind the restriction of the ket and bra on $\chi_u$. In the latter two computations, the environments disagree either due to population or state, thus the outside ket and bra do not coincide on the complement of $\chi_u$ and this takes the entire term to zero.

We now return to our more complex example, where states of some auxiliary nodes are interpreted by $\harp{\chi}$ as edges. Consider the state of equation $(1)$, in this case the $\harp{\chi}_3^{1}$-traceout returns a probabilistic mixture of two causal environments:
\[\includegraphics[scale = 0.2]{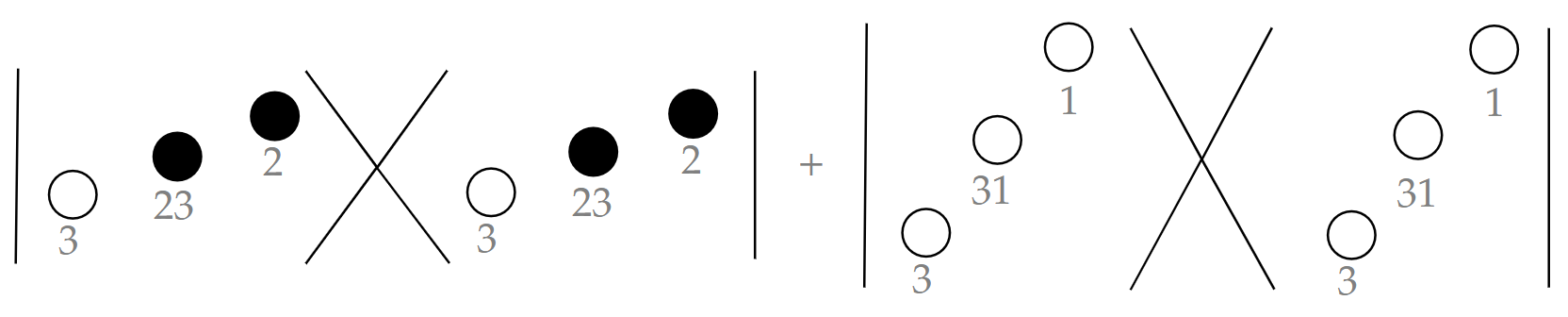}\] 
Now, the tensor product corresponding to a restriction $\chi$ works by weaving a restricted configuration $\ket{G_{\chi}}$ and its complement $\ket{G_{\overline{\chi}}}$ back together, so long as they are consistent:
\begin{definition}[Tensor]\label{def:tensors}

Every restriction $\chi$ induces a tensor $\tensorchi:{\cal H}\times{\cal H}\rightarrow {\cal H}$ defined by:\\
$ \ket{H} \tensorchi \ket{H'} :=\begin{cases}
\ket{G} & \text{when } H=G_{\chi } ,H'=G_{\overline{\chi }}\\
0 & \text{otherwise}
\end{cases}$\\
$ \ket{\psi } \tensorchi \ket{\psi '}$ is defined from the above, by bilinear extension.\\
$ \ket{G}\bra{H} \tensorchi \ket{G'}\bra{H'} :=\left(\ket{G} \tensorchi \ket{G'}\right)\left(\bra{H} \tensorchi \bra{H'}\right)$.\\
For any two operators $ A,B$, we define $ A\tensorchi B$ from the above by bilinear extension.
\end{definition}

Notice how textbook tensors work go from go from ${\cal H}_A\times {\cal H}_B$ to ${\cal H}_{AB}$, i.e. composing two different Hilbert spaces into a third. Here the approach is different; we really think of them as decomposing elements of the large Hilbert space into pieces that still belong the large Hilbert space. This internalizes the tensor, which becomes a binary operator of ${\cal H}$.

Let us give example calculations in terms of $\chi_u$:  \[   \includegraphics[scale = 0.25]{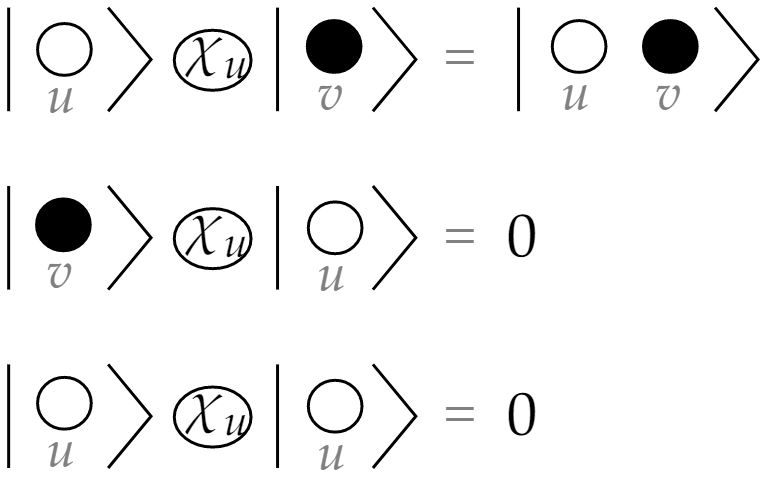} \]  In the top calculation, the pair $u,v$ can  be understood as weaving by $\chi_u$ with $u$ singled out and $v$ left behind. For the middle, $\chi_u$ does not ever single out $v$, so the product with $v$ on the left is $0$. Finally, on the bottom line, $\chi_u$ does not ever leave behind $u$ so the product must give $0$.

Another example based on the $\mu$-tensor shows that the splitting between subsystems is now quantized:\[\includegraphics[scale = 0.25]{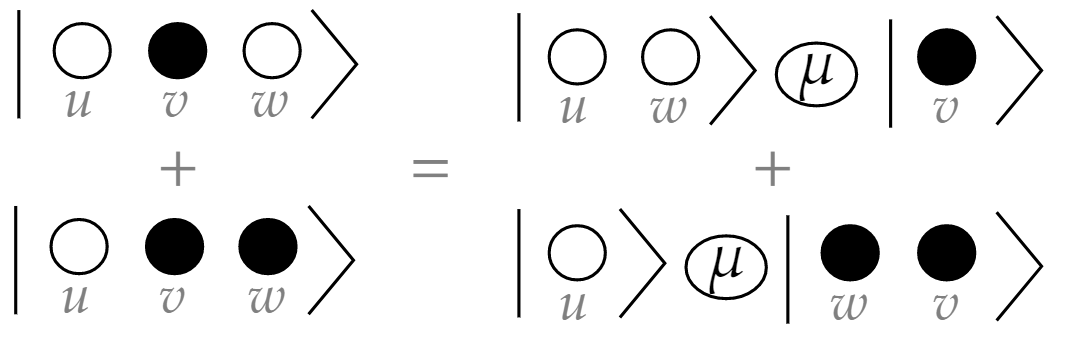} \]
i.e. the number of systems on the left depends on the branch of the superposition. The same example illustrates the fact that entanglement is a restriction-dependent concept, as the above $\mu$-entangled state is in fact a $\chi_{uv}$-product state: \[\includegraphics[scale = 0.2]
{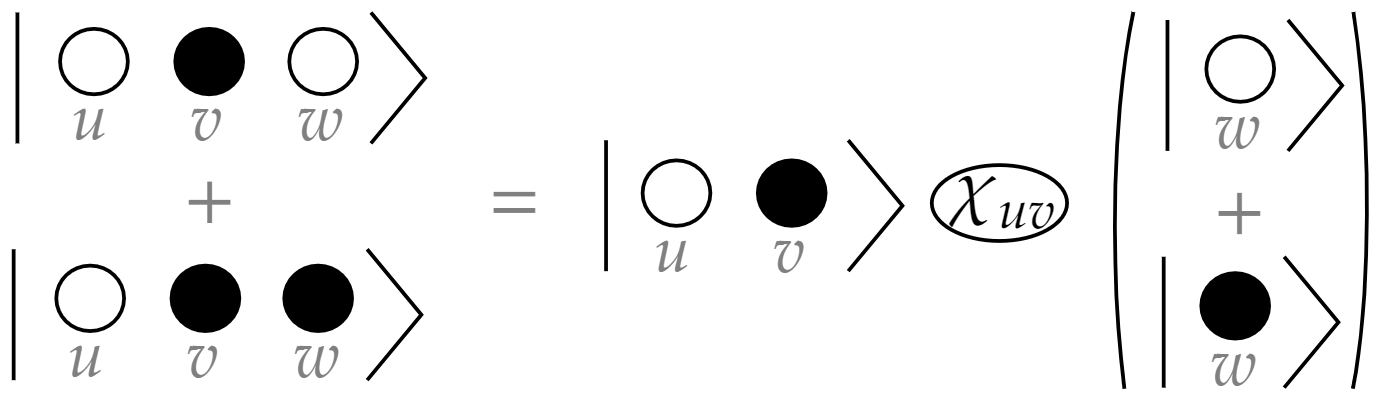}\]




We now consider classes of states which are compatible with each other, i.e. $\ket{\psi}$ and $\ket{\psi'}$ are $\chi$-consistent if whenever $\ket{\phi}$ contains configuration $G$ and $\ket{\phi'}$ contains configuration $G'$ then $G$ should be given by $G = \chi (G \cup G')$. An important class of operators for us will be those which do not generate inconsistency within states. 
\begin{definition}[Consistency]
$ \ket{\psi } ,\ket{\psi '}$ are $ \chi $-consistent if and only if $ \braket{G| \psi }\braket{G'|\psi '} \neq 0$ implies $ \ket{G} \tensorchi \ket{G'} \neq 0$.\\ 
$ \rho ,\sigma $ are $ \chi $-consistent if and only if $ \rho _{GH} \sigma _{G'H'} \neq 0$ implies $ \ket{G} \tensorchi \ket{G'} \neq 0\neq \ket{H} \tensorchi \ket{H'}$, where $ \rho _{GH} :=\bra{G} \rho \ket{H}$. An operator $A$ is $ \chi $-consistent-preserving if and only if $ \bra{H} A\ket{G_{\chi }} \neq 0$ entails $ \ket{H} \tensorchi \ket{G_{\overline{\chi }}} \neq 0$, and $ \bra{H} A^{\dagger }\ket{G_{\chi }} \neq 0$ entails $ \ket{H} \tensorchi \ket{G_{\overline{\chi }}} \neq 0$. 
\end{definition}

Using generalised partial traces and tensors sometimes feels like a step into the unknown. Our old intuitions about traceouts and tensors guide us, but in some occasions they mislead us. In the long version of this paper [F] we have, as a result, checked all conditions of applications of several basic facts about the way these tensor and tracing operators interact with one another, leading to the Toolbox of table \ref{tab:toolbox}.

A good rule of thumb is that usual intuitions about $A\tensorchi B$ will carry through provided that $\chi $-consistency conditions are met. In fact much of the attention in the proofs is spent keeping track of which terms get zeroed by the function $\tensorchi$. Another good rule of thumb is that our usual intuitions about subsystems $\zeta $ of a wider system $\chi $ will carry through, provided that the comprehension condition is met, which is true of most natural cases, especially under name-preservation conditions, see Prop. [F].2. In appendix \ref{app:props} we show that the partial trace is completely positive up to extensions by generalized tensors.


\paragraph*{Properties} 
An early attempt to define a (non-modular) partial trace for quantum causal configuration dynamics actually failed to exhibit positivity-preservation, i.e. there exists $\rho $ non-negative with $\rho _{|\chi }$ not non-negative \cite{ArrighiQCGD}. 
Here we show that partial traces are positive-preserving. In fact we check they are completely-positive-preserving, meaning that they remain positive-preserving when tensored with the identity, as required for general quantum operations. We do the same for trace-preservation. 

We denote by $(( .)_{|\chi }  \tensorzeta   I)$ the map $\rho  \tensorzeta   \sigma \ \mapsto \rho _{|\chi }  \tensorzeta   \sigma $ linearly extended to the whole of $\mathcal{H}$.

\begin{proposition}[Traceouts positivity-preservation and trace-preservation]\label{prop:traceouts}
The map \ $ \rho \mapsto (( .)_{|\chi }  \tensorzeta   I)( \rho )$ over trace class operators is positive-preserving.\\
If moreover $ \ket{G_{\zeta \chi }}  \tensorzeta   \ket{G_{\overline{\zeta }}} \neq 0$, then the same map is trace-preserving.
\end{proposition}

\section{Locality}\label{sec:locality}
Since a restriction $\chi$ isolates a part of each possible configuration, one can introduce the notion of a $\chi $-local operator, one that only acts on the restriction $\chi $, leaving its complement $\overline{\chi }$ unchanged. I.e. a $\chi $-local operator acts only on the left of $ \tensorchi $.

\begin{figure}\centering
\includegraphics[width=\textwidth]{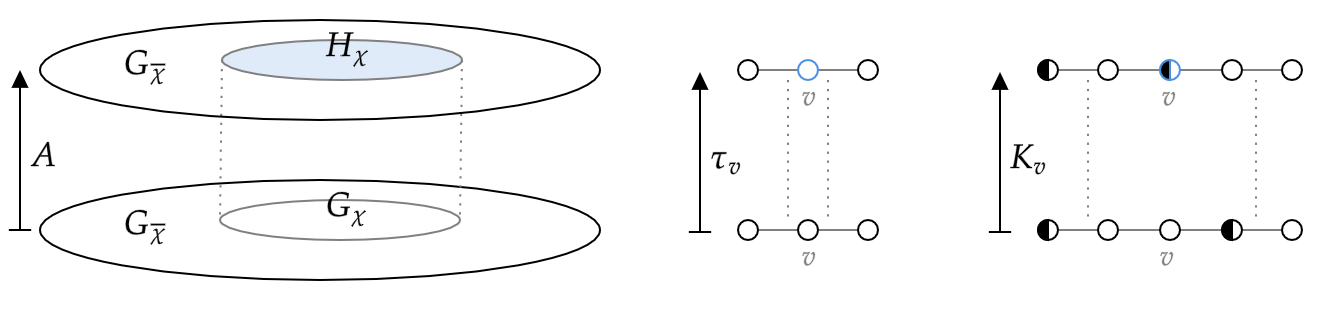}
\caption{\label{fig:locality}{\em Local operators.} Left: $A\ \chi $-local will only modify $G_{\chi }$. Middle: $\zeta _{v}$-local operator $\tau _{v}$ just toggles a 'black/blue' bit inside the system at $v$. Right: $\chi _{v}$-local operator $K_{v}$ is a reversible update rule which computes the future state of the system at $v$ (according to $M$, see Fig. \ref{fig:causality}) and toggles it out, whilst attempting to leave the rest mostly unchanged, cf. Th. \ref{th:blockdecomposition}.}
\end{figure}

\begin{definition}[Locality]\label{def:locality}

$A$ is $\chi $-local if and only if 
\begin{equation}
\bra{H} A\ket{G} =\bra{H_{\chi }} A\ket{G_{\chi }}\braket{H_{\overline{\chi }}|G_{\overline{\chi }}}\label{eq:locality}
\end{equation}
\end{definition}


The standard way to state the locality of $A$ is to write it as $A=B\otimes I$. Here follows a generalization of this statement, a key point that allows for this generalisation is that tensoring with the identity zeroes out the non-local terms of $B$.

\begin{proposition}[Operational locality]\label{prop:gatelocality}
$A$ is $\chi $-local if and only if $A=A\tensorchi I$.\\
For all $ B$, $B\tensorchi I$ is $\chi $-local.\\
\end{proposition}
The operational picture as presented here is in fact a constructive tool, \textit{any} operator $A: \mathcal{G}_{\chi} \rightarrow \mathcal{G}_{\chi}$ on the image $\mathcal{G}_{\chi}$ of $\chi$ on $\mathcal{G}$ can be extended to a local operator on $\mathcal{G}$ by computing $\bar{A} := A \tensorchi I$, this always works since $(A \tensorchi I) \tensorchi I = A \tensorchi I$. Whilst locality captures the intuition that a $\chi$-local operator $A$ affects only the $\chi$ part of a configuration, it turns out that some local operators may introduce inconsistencies between the part $\chi$ they act on, and their surroundings. This motivates the following definition:
\begin{definition}[Strict locality]
$A$ is strictly $ \chi $-local if and only if $ A$ is $ \chi $-local and $ \chi $-consistent-preserving.
\end{definition}
Actually, the fact that an operator $A$ may sometimes be $\chi $-local but not strictly $\chi $-local, comes from the fact we allow for operators that lose norm, see Fig. \ref{fig:strictlocality}. Indeed, $A$ may change the state of $\ket{G_{\chi }}$, in such a way that makes it inconsistent with $\ket{G_{\overline{\chi }}}$, resulting in a loss of norm. None of these issues arise if $A$ is unitary.
\begin{proposition}\label{prop:strictlocality}
$A$ is strictly $\chi $-local if and only if, $A^{\dagger } A$ and $AA{^{\dagger }}^{\ }$ are $\chi $-local.\\ 
In particular, every unitary $\chi $-local is strictly $\chi $-local.
\end{proposition}

The definition of strict $\chi$-locality is also motivated by the fact that if $A,B$ are strictly $\chi$-local operators then the product $AB$ is strictly $\chi$-local. This is a direct consequence of Prop. \ref{prop:gatelocality} and Lem. [F].11. Conversely, the fact that local but not strictly local operators are not composable by is a direct consequence of Prop. \ref{prop:strictlocality}. An important corollary for the soundness of the notion of locality is that unitary local operators are always composable.

\begin{figure}\centering
\includegraphics[width=0.3\textwidth]{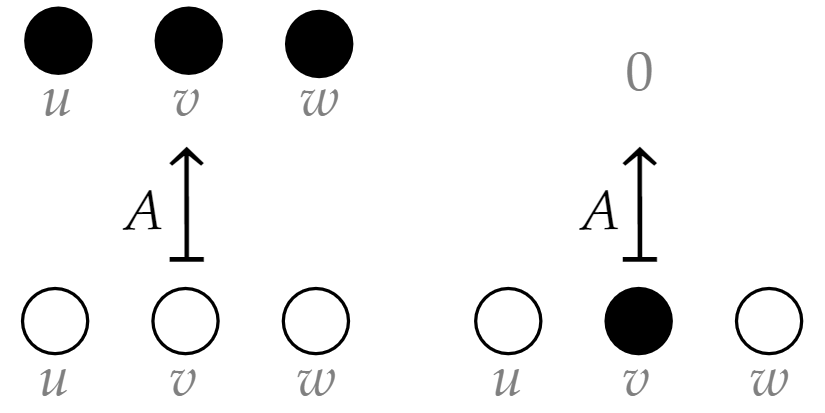}
\caption{\label{fig:strictlocality}{\em Strict locality.} Consider a restriction $\chi:G \rightarrow G_{\chi}$ given by (i) If $G$ contains only black nodes $G_{\chi} =  G$ else (ii) $G_{\chi}$ is the subconfiguration of $G$ containing only its white nodes. Let $\mathbf{flip}(G)$ be the configuration with $white.u \in \mathbf{flip}(G) \Leftrightarrow black.u \in G$. Consider operator $A$ such that (i) $A\ket{G} = 0$ if $G$ has any black nodes and (ii) $A\ket{G} = \ket{\mathbf{flip}(G)}$ otherwise, that is when all nodes are white. $A$ is local but not consistency-preserving. Indeed, $\ket{white.u}$ and $\ket{black.v}$ are consistent since $white.u = \chi(white.u \cup black.v)$, but $A\ket{white.u} = \ket{black.u}$ and $\ket{black.v}$ are not consistent since $\chi(black.u \cup black.v) = black.u \cup black.v \neq black.u$.}
\end{figure}


Locality can furthermore be phrased in this Heisenberg picture, the result of a $\chi $-local observable on $\rho $ solely depends on its partial trace $\rho _{|\chi }$.
\begin{proposition}[Dual locality]\label{prop:duallocality}

$A$ is $\chi $-local if and only if $ ( A \rho )_{|\emptyset } =( A \rho _{|\chi })_{|\emptyset }$.
\end{proposition}
The above proposition states that $\rho _{|\chi }$ contains the part of $\rho $ that is observable by $\chi $-local operators. The next proposition states that $\rho _{|\chi }$ does contain anything more.

\begin{proposition}[Local tomography]\label{prop:tomography}
If for all $A$ $\chi $-local $ ( A\rho )_{|\emptyset } =( A\sigma )_{|\emptyset }$ then $ \rho _{|}{}_{\chi } =\sigma _{|\chi }$. \\
Moreover, if $ \rho $, $ \sigma $ are name-preserving and for all $A$ $\chi $-local and name-preserving $ ( A\rho )_{|\emptyset } =( A\sigma )_{|\emptyset }$ then $ \rho _{|}{}_{\chi } =\sigma _{|\chi }$. 
\end{proposition}


\section{Causality}\label{sec:causality}

\begin{figure}\centering
\includegraphics[width=\textwidth]{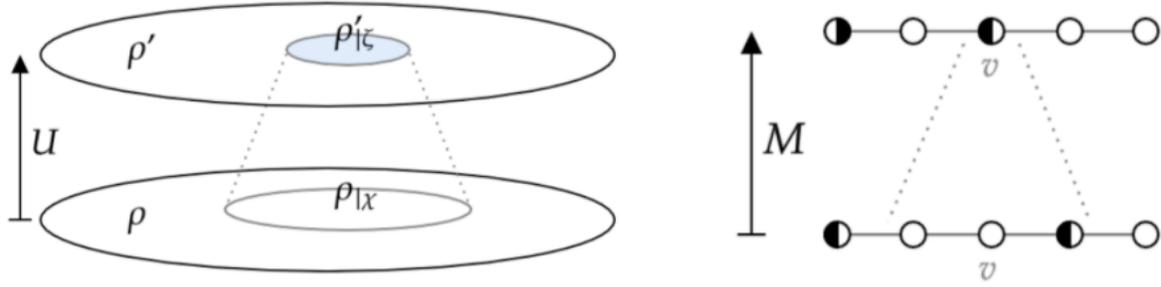}
\caption{\label{fig:causality}{\em Causal operators.} Left: $U\ \chi \zeta $-causal may modify the whole of $\rho $. But it is such that $\rho '_{|\zeta }$ solely depends on $\rho _{|\chi }$. Right: $\chi _{v} \zeta _{v}$-causal operator $M$ propagates particles. They bounce on borders.}
\end{figure}

\begin{figure}\centering
\includegraphics[scale = 0.2]{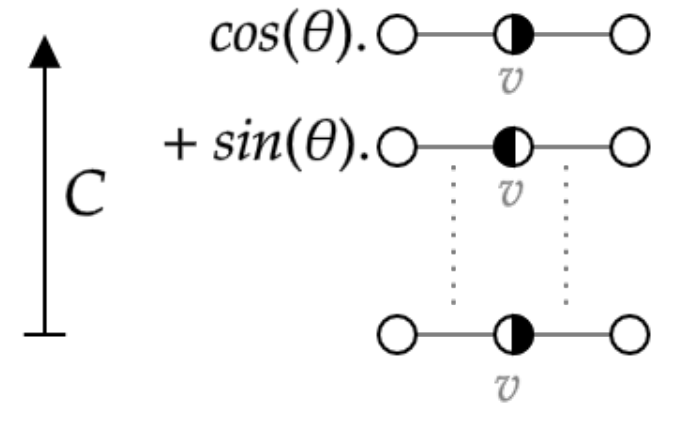}
\caption{\label{fig:quantumcausality} {\em Causal operators yielding superpositions of states.} Left: Instead of simply iterating $M$, we could iterate $MC$, where $C$ acts on every right-moving (resp. left-moving) particle by placing it in a superposition of being right-moving with amplitude $cos( \theta )$ and left-moving with amplitude $sin( \theta )$ (resp. left-moving with amplitude $cos( \theta )$ and right-moving with amplitude $-sin( \theta )$.}
\end{figure}
A key principle in physical frameworks is the propagation of information, in standard quantum information-theoretic settings causality constraints can be imposed using traces \cite{Beckman}, here we generalise this approach to our definition of trace.
Consider two restrictions $\chi,\zeta$ over configurations, a $\chi \zeta $-causal operator is one which restricts information propagation by imposing that region $\zeta $ at the next time step depends only upon region $\chi $ at the previous time step. Subject to this constraint, $\chi \zeta $-causal operator will be permitted to edit the entirety of the configurations they act on.


\begin{definition}[Causality]\label{def:causality}

$U$ is $\chi \zeta $-causal if and only if
\begin{equation}
\left( U\rho U^{\dagger }\right)_{|\zeta } =\left( U\rho _{|\chi } U^{\dagger }\right)_{|\zeta }\label{eq:causalitly}
\end{equation}
\end{definition}
Directly from then definition it follows that causality constraints are composable, that is, if $U$ is $\chi \zeta$-causal and $V$ is $\zeta \eta$ -causal, then $VU$ is $\chi\eta$-causal.
In the Heisenberg picture, it turns out that $\chi \zeta $-causality actually states that whatever can be $\zeta $-locally observed at the next time step, could be $\chi $-locally observed at the previous time step.

\begin{proposition}[Dual causality]\label{prop:dualcausality}
$U$ is $\chi \zeta $-causal if and only if for all $ A$ $ \zeta $-local, $ U^{\dagger } AU$ is $ \chi $-local.\\
\end{proposition}


Causality is a basic Physics principle, anchored on the postulate that information-propagation is bounded by the speed of light. Yet causality is a top-down axiomatic constraint. When modelling an actual Physical phenomenon, we need a bottom-up, constructive way of expressing the dynamics. We usually proceed by describing it in terms of local interactions, happening simultaneously and synchronously. The following shows that causal operators are always of that form, under the restriction that the evolution be name-preserving.
\begin{definition}[Name-preservation]\label{def:np}
Let $A$ be an operator over configurations. It is said to be name-preserving (n.-p. for short) if and only if $V(G) \neq V( H)$ implies $\bra{H} A\ket{G} =0$.
\end{definition}

\begin{thm}[Block decomposition]\label{th:blockdecomposition}
Let $ \zeta _{v}$ be the pointwise restriction such that \[ \zeta _{v}(\{\sigma '.u\}) :=\begin{cases}
\{\sigma '.u\} & \text{if} \ u=v\\
\emptyset  & \text{otherwise}
\end{cases}\].\\
Consider $ U$ a name-preserving unitary operator over $ \mathcal{H}$, which for all $ v\in \mathcal{V}$ is $ \chi _{v} \zeta '_{v}$-causal, with $ \zeta _{v} \sqsubseteq \zeta '_{v}$ and $\chi_v$ extensible. Let $ \Sigma '=\{0,1\} \times \Sigma $, let $ \mathcal{G} '$ and let $ \mathcal{H} '$ be the corresponding Hilbert space. \\
Let $ \mu $ be the pointwise restriction such that $ \mu (\{b.\sigma .u\}) :=\begin{cases}
\{0.\sigma .u\} & \text{if} \ b=0\\
\emptyset  & \text{otherwise}
\end{cases} \ $. \\
Let $ \mathcal{G} '$ be the set of finite subsets of $ \mathcal{S} ':=\Sigma '\times \mathcal{V}$ and $ \mathcal{H} '$ the Hilbert space whose canonical o.n.b is $ \mathcal{G} '$. Over $ \mathcal{H} '$, there exists $ \tau _{v}$ a strictly $ \zeta _{v}$-local unitary and $ K_{v}$ a strictly $ \xi _{v}$-local unitary such that
\begin{equation*}
\begin{aligned}
\forall \ket{\psi } \in \mathcal{H} '_{\mu } \cong \mathcal{H} ,\ \left(\prod _{v\ \in \ \mathcal{V}} \tau _{v}\right)\left(\prod _{v\ \in \ \mathcal{V}} K_{v}\right)\ket{\psi } & =U\ket{\psi }
\end{aligned}
\end{equation*}
where $ \xi _{v} :=\mu \chi _{v} \cup \overline{\mu } \zeta _{v}$. In addition,$ \ [ K_{x} ,K_{y}] =[ \tau _{x} ,\tau _{y}] =0$. 
\end{thm}

Such causal-to-local theorems are renowned to be difficult and follow a long tradition starting with just two systems \cite{Beckman}; moving on to three \cite{SchumacherWestmoreland}; a line of them \cite{SchumacherWerner}; static networks \cite{ArrighiUCAUSAL}, and finally node-preserving but connectivity-varying networks \cite{ArrighiQCGD} a.k.a `quantum causal graph dynamics'. The point of the above theorem is that the decomposition carries through to arbitrary restrictions $\chi_v$. From a methodological point of view, we used this theorem a test bench, to make sure that we had put together a set of mathematical tools that would be sufficient to combine and establish non-trivial results.

\section{Perspectives}\label{sec:conclusion}

\color{black}
We ran into generalised tensors and traces whilst seeking to provide rigorous kinematics and fully quantum evolutions over quantum superpositions of networks, equipped with rigorous notions of locality and causality. This was for the sake of taking networks models of complex systems, into the quantum realm. This is achieved in the full version [F] and could be useful for instance, in the field of Quantum Gravity. 
Similarly, in the field of Quantum Computing, we now have a framework in which to model fully-quantum distributed computing devices, possibly including dynamics over indefinite causal orders \cite{BruknerDynamics} by means of causal unitary operators over superpositions of directed networks. But generalised of tensors and traceouts may lead to a whole range of byproducts:
\begin{description}
\item {\em Towards base-independence}. To make the notion of subsystem base-independent, algebraic quantum field theory tends to think of them as Von Neumann algebras instead \cite{BratteliRobinson}\cite{gogosiofunc}. This approach has been formalised in \cite{GogiosoChurch} in the the setting of categorical quantum mechanics, where the tensor product is a seen as a binary operator between states associated to commuting algebras. Such tensor products are therefore partially defined and remain abstract mathematical objects. The generalised tensor products of the present paper are everywhere defined and constructive, but they are base-dependent. It would be interesting to bridge the gap between these two notions to get the best of both worlds. One way to go about this is to use the Wedderburn-Artin theorem. 
\item {\em Decomposition techniques, causal-to-local}. These generalized operators were essential to the theorem representing causal unitary operators by means of local unitary gates. Many variants of this question are still open however, even for static networks over a handful of systems \cite{Beckman,SchumacherWestmoreland, Lorenz_2021}, as soon as we demand that the representation be exact. Recent approaches \cite{VanrietveldeRouted} to phrasing the answers to these questions make the case for annotating wires with a type system specifying which subspace will flow into them; this in turn has the flavour of a generalized tensor product. This suggests that the generalized operators, by means of their increased expressiveness, may be key to reexpress and prove a number of standing conjectures.
\item {\em Construction techniques, local-to-causal.} In \cite{ArrighiAQG} the authors provide a hands-on, concrete way of expressing a family of unitary evolutions over network configurations allowing for quantum superpositions of connectivities. One may wonder whether these addressable quantum gates, if extended to become able to split and merge, could be proven universal in the class of causal operators over quantum networks. 
\item {\em Fusion products between parts of constrained configuration spaces}. In most physical theories, the set of allowed configurations is constrained. For instance, charges and fields are constrained by the Gauss law. It follows that the tensor product between two regions of space may be ill-defined, for instance because both of them follow the Gauss law at the individual level, but not when they are placed next to one another. The hereby devised generalised tensor product is robust enough to handle these situations whilst preserving most desired algebraic properties, simply by sending them to the null vector. The gauge-invariant `fusion product' of \cite{FreidelFusion} focusses upon this issue in the continuum by following a different, gauge-invariant and partially-defined approach.
\item {\em Flexible notions of entanglement, between logical spaces}. The generalized tensor $\tensorchi$ allows us to define the entanglement between its the left factor and the right factor according to an almost arbitrary logical criterion $\chi$. For a bipartite pure state we may take the Von Neumann entropy of its partial trace on $\chi$ in order to quantify this entanglement. 
\item {\em Modelling delocalized observers, quantum reference frames}. Decoherence theory \cite{PazZurek} models the observer as a quantum system interacting with others; and the post-measurement state as that obtained by tracing out the observer. But since the observer is quantum, it could be delocalized, raising the question of what it means to take the trace out then. Here we can model a delocalized observer by delocalized black particles, and trace them out. This ability to ``take the vantage point of delocalized quantum system'' is in fact a feature in common with quantum reference frames. 
\item {\em Ad hoc notions of causality, emergence of space}. The notion of $\chi\zeta$-causality allows us to define causality constraints according to almost arbitrary families of logical criteria $(\chi,\zeta)$. This includes scenarios where all black particles communicate whatever their network distance, say. In fact the very notion of network connectivity is arbitrary in the theory, i.e. $\zeta^r$ can in principle be redefined in order to better fit ad hoc causality constraints, possibly emerging in a similar way to pointer states in decoherence theory \cite{PazZurek}. 
\end{description}
\bibliography{biblio}

\begin{thebibliography}{10}

\bibitem{ArrighiQCGD}
P.~Arrighi and S.~Martiel.
\newblock Quantum causal graph dynamics.
\newblock {\em Physical Review D.}, 96(2):024026, 2017.
\newblock Pre-print arXiv:1607.06700.

\bibitem{ArrighiUCAUSAL}
P.~Arrighi, V.~Nesme, and R.~Werner.
\newblock {Unitarity plus causality implies localizability}.
\newblock {\em J. of Computer and Systems Sciences}, 77:372--378, 2010.
\newblock Local proceedings of QIP 2010 (long talk).

\bibitem{ArrighiAQG}
Pablo Arrighi, Christopher Cedzich, Marin Costes, Ulysse R{\'e}mond, and
  Beno{\^\i}t Valiron.
\newblock Addressable quantum gates.
\newblock {\em arXiv preprint arXiv:2109.08050}, 2021.

\bibitem{ArrighiQNT}
Pablo Arrighi, Am{\'e}lia Durbec, and Matt Wilson.
\newblock Quantum networks theory.
\newblock {\em arXiv preprint arXiv:2110.10587}, 2021.

\bibitem{Spekkens_refs}
Stephen~D. Bartlett, Terry Rudolph, and Robert~W. Spekkens.
\newblock Reference frames, superselection rules, and quantum information.
\newblock {\em Rev. Mod. Phys.}, 79:555--609, Apr 2007.
\newblock URL: \url{https://link.aps.org/doi/10.1103/RevModPhys.79.555}, \href
  {https://doi.org/10.1103/RevModPhys.79.555}
  {\path{doi:10.1103/RevModPhys.79.555}}.

\bibitem{Beckman}
D.~Beckman, D.~Gottesman, M.~A. Nielsen, and J.~Preskill.
\newblock {Causal and localizable quantum operations}.
\newblock {\em Phys. Rev. A.}, 64(052309), 2001.

\bibitem{BratteliRobinson}
O.~Bratteli and D.~Robinson.
\newblock {\em {Operators algebras and quantum statistical mechanics 1}}.
\newblock Springer, 1987.

\bibitem{BruknerDynamics}
Esteban Castro-Ruiz, Flaminia Giacomini, and \ifmmode
  \check{C}\else~\v{C}\fi{}aslav Brukner.
\newblock Dynamics of quantum causal structures.
\newblock {\em Phys. Rev. X}, 8:011047, Mar 2018.
\newblock URL: \url{https://link.aps.org/doi/10.1103/PhysRevX.8.011047}, \href
  {https://doi.org/10.1103/PhysRevX.8.011047}
  {\path{doi:10.1103/PhysRevX.8.011047}}.

\bibitem{ManyWorlds}
Kostia Chardonnet, Marc de~Visme, Beno{\^\i}t Valiron, and Renaud Vilmart.
\newblock The many-worlds calculus.
\newblock {\em arXiv preprint arXiv:2206.10234}, 2022.

\bibitem{chiribella_subsystems}
Giulio Chiribella.
\newblock Agents, subsystems, and the conservation of information.
\newblock {\em Entropy}, 20(5):358, May 2018.
\newblock URL: \url{http://dx.doi.org/10.3390/e20050358}, \href
  {https://doi.org/10.3390/e20050358} {\path{doi:10.3390/e20050358}}.

\bibitem{QCircuitCompleteness}
Alexandre Cl{\'e}ment, Nicolas Heurtel, Shane Mansfield, Simon Perdrix, and
  Benoit Valiron.
\newblock A complete equational theory for quantum circuits.
\newblock {\em arXiv preprint arXiv:2206.10577}, 2022.

\bibitem{Sheets}
Cole Comfort, Antonin Delpeuch, and Jules Hedges.
\newblock Sheet diagrams for bimonoidal categories.
\newblock {\em arXiv preprint arXiv:2010.13361}, 2020.

\bibitem{hamette_qrt}
Anne-Catherine de~la Hamette and Thomas~D. Galley.
\newblock Quantum reference frames for general symmetry groups.
\newblock {\em Quantum}, 4:367, Nov 2020.
\newblock URL: \url{http://dx.doi.org/10.22331/q-2020-11-30-367}, \href
  {https://doi.org/10.22331/q-2020-11-30-367}
  {\path{doi:10.22331/q-2020-11-30-367}}.

\bibitem{FreidelFusion}
William Donnelly and Laurent Freidel.
\newblock Local subsystems in gauge theory and gravity.
\newblock {\em Journal of High Energy Physics}, 2016(9):1--45, 2016.

\bibitem{moliner_space}
Pau Enrique~Moliner, Chris Heunen, and Sean Tull.
\newblock Space in monoidal categories.
\newblock {\em Electronic Proceedings in Theoretical Computer Science},
  266:399–410, Feb 2018.
\newblock URL: \url{http://dx.doi.org/10.4204/EPTCS.266.25}, \href
  {https://doi.org/10.4204/eptcs.266.25} {\path{doi:10.4204/eptcs.266.25}}.

\bibitem{GogiosoChurch}
Stefano Gogioso.
\newblock A process-theoretic church of the larger hilbert space.
\newblock {\em arXiv preprint arXiv:1905.13117}, 2019.

\bibitem{gogosiofunc}
Stefano Gogioso, Maria~E. Stasinou, and Bob Coecke.
\newblock Functorial evolution of quantum fields.
\newblock {\em Frontiers in Physics}, 9:24, 2021.
\newblock URL:
  \url{https://www.frontiersin.org/article/10.3389/fphy.2021.534265}, \href
  {https://doi.org/10.3389/fphy.2021.534265}
  {\path{doi:10.3389/fphy.2021.534265}}.

\bibitem{Lorenz_2021}
Robin Lorenz and Jonathan Barrett.
\newblock Causal and compositional structure of unitary transformations.
\newblock {\em Quantum}, 5:511, Jul 2021.
\newblock URL: \url{http://dx.doi.org/10.22331/q-2021-07-28-511}, \href
  {https://doi.org/10.22331/q-2021-07-28-511}
  {\path{doi:10.22331/q-2021-07-28-511}}.

\bibitem{PazZurek}
J.~P. Paz and W.~H. Zurek.
\newblock {Environment-induced decoherence and the transition from quantum to
  classical}.
\newblock {\em Lecture Notes in Physics}, pages 77--140, 2002.

\bibitem{RovelliLQG}
Carlo Rovelli.
\newblock Simple model for quantum general relativity from loop quantum
  gravity.
\newblock In {\em Journal of Physics: Conference Series}, volume 314, page
  012006. IOP Publishing, 2011.

\bibitem{SchumacherWerner}
B.~Schumacher and R.~Werner.
\newblock {Reversible quantum cellular automata.}
\newblock arXiv pre-print quant-ph/0405174, 2004.

\bibitem{SchumacherWestmoreland}
B.~Schumacher and M.~D. Westmoreland.
\newblock {Locality and information transfer in quantum operations}.
\newblock {\em Quantum Information Processing}, 4(1):13--34, 2005.

\bibitem{VanrietveldeRouted}
Augustin Vanrietvelde, Hl{\'e}r Kristj{\'a}nsson, and Jonathan Barrett.
\newblock Routed quantum circuits.
\newblock {\em Quantum}, 5:503, 2021.

\bibitem{ConsistentICO}
Augustin Vanrietvelde, Nick Ormrod, Hl{\'e}r Kristj{\'a}nsson, and Jonathan
  Barrett.
\newblock Consistent circuits for indefinite causal order.
\newblock {\em arXiv preprint arXiv:2206.10042}, 2022.

\bibitem{zanardi_virtual}
Paolo Zanardi.
\newblock Virtual quantum subsystems.
\newblock {\em Physical Review Letters}, 87(7), Jul 2001.
\newblock URL: \url{http://dx.doi.org/10.1103/PhysRevLett.87.077901}, \href
  {https://doi.org/10.1103/physrevlett.87.077901}
  {\path{doi:10.1103/physrevlett.87.077901}}.

\bibitem{zanardi_observe}
Paolo Zanardi, Daniel~A. Lidar, and Seth Lloyd.
\newblock Quantum tensor product structures are observable induced.
\newblock {\em Phys. Rev. Lett.}, 92:060402, Feb 2004.
\newblock URL: \url{https://link.aps.org/doi/10.1103/PhysRevLett.92.060402},
  \href {https://doi.org/10.1103/PhysRevLett.92.060402}
  {\path{doi:10.1103/PhysRevLett.92.060402}}.

\end{thebibliography}

\appendix

\section{Proofs of propositions}\label{app:props}

\paragraph*{Proof of Prop. \ref{prop:traceouts}}

\begin{proof} 

[Complete positivity preservation]

\begin{align*}
\ket{G}\bra{H} & =\ket{G_{\zeta }}\bra{H_{\zeta }}  \tensorzeta   \ket{G_{\overline{\zeta }}}\bra{H_{\overline{\zeta }}}\\
(( .)_{|\chi }  \tensorzeta   I) \ket{G}\bra{H} & =\ket{G_{\zeta \chi }}\bra{H_{\zeta \chi }}\braket{H_{\zeta \overline{\chi }}|G_{\zeta \overline{\chi }}}  \tensorzeta   \ket{G_{\overline{\zeta }}}\bra{H_{\overline{\zeta }}}
\end{align*}

Let $\alpha '_{G_{\chi } KG'} :=\begin{cases}
\alpha _{GG'} & \text{if} \ \ket{G_{\chi }} \tensorchi \ket{K} \ =\ \ket{G}\\
0 & \text{otherwise}
\end{cases}$ .
\begin{align*}
\textrm{Consider some }\ket{\psi } & =\sum _{G,\ G'\ \in \ \mathcal{G}} \alpha _{G G'} \ \ket{G}  \tensorzeta   \ket{G'}\\
\bra{\psi } & =\sum _{H,\ H'\ \in \ \mathcal{G}} \alpha _{HH'}^{*} \ \bra{H}  \tensorzeta   \bra{H'}\\
\ket{\psi }\bra{\psi } & =\sum _{G,\ G',\ H,\ H'\ \in \ \mathcal{G}} \alpha _{G G'} \alpha _{HH'}^{*} \ \ket{G}\bra{H}  \tensorzeta   \ket{G'}\bra{H'}\\
(( .)_{|\chi }  \tensorzeta   I)\left(\ket{\psi }\bra{\psi }\right) & =\sum _{ \begin{array}{l}
G,\ G',\ H,\ H'\ \in \ \mathcal{G}\\
\ket{G}  \tensorzeta   \ket{G'} \ \neq \ 0\\
\ket{H}  \tensorzeta   \ket{H'} \ \neq \ 0
\end{array}} \alpha _{G G'} \alpha _{HH'}^{*} \ \ket{G_{\chi }}\bra{H_{\chi }}\braket{H_{\overline{\chi }}|G_{\overline{\chi }}}  \tensorzeta   \ket{G'}\bra{H'}\\
 & =\sum _{ \begin{array}{l}
G ,\ H ,\ G',\ H'\ \in \ \mathcal{G}\\
G_{\overline{\chi }} =H_{\overline{\chi }}\\
\ket{G}  \tensorzeta   \ket{G'} \ \neq \ 0\\
\ket{H}  \tensorzeta   \ket{H'} \ \neq \ 0
\end{array}} \alpha _{G G'} \alpha _{HH'}^{*} \ \ket{G_{\chi }}\bra{H_{\chi }}  \tensorzeta   \ket{G'}\bra{H'}\\
 & =\sum _{ \begin{array}{l}
G_{\chi } ,\ H_{\chi } ,\ K,\ G',\ H'\ \in \ \mathcal{G}\\
\left(\ket{G_{\chi }} \tensorchi \ket{K}\right)  \tensorzeta   \ket{G'} \ \neq \ 0\\
\left(\ket{H_{\chi }} \tensorchi \ket{K}\right)  \tensorzeta   \ket{H'} \ \neq \ 0
\end{array}} \alpha '_{G_{\chi } KG'} \alpha ^{\prime *}_{H_{\chi } KH'} \ \ket{G_{\chi }}\bra{H_{\chi }}  \tensorzeta   \ket{G'}\bra{H'}\\
\ket{\phi ^{K}} & :=\sum _{ \begin{array}{l}
G_{\chi } ,\ G'\ \in \ \mathcal{G}\\
\left(\ket{G_{\chi }} \tensorchi \ket{K}\right)  \tensorzeta   \ket{G'} \ \neq \ 0
\end{array}} \alpha '_{G_{\chi } KG'}\ket{G_{\chi }}  \tensorzeta   \ket{G'}\\
(( .)_{|\chi }  \tensorzeta   I)\left(\ket{\psi }\bra{\psi }\right) & =\sum _{K\ \in \ \mathcal{G}}\ket{\phi ^{K}}\bra{\phi ^{K}}\\
(( .)_{|\chi }  \tensorzeta   I)\left(\sum _{i}\ket{\psi ^{i}}\bra{\psi ^{i}}\right) & =\sum _{K\ \in \ \mathcal{G} ,\ i}\ket{\phi ^{i,K}}\bra{\phi ^{i,K}}
\end{align*}

[Trace Preservation]

Notice that $\overline{\zeta \chi \cup \overline{\zeta }} =\zeta \overline{\chi }$.
\begin{align*}
(( .)_{\chi }  \tensorzeta   I)\left(\ket{G}\bra{H}\right) & =\ \left(\ket{G_{\zeta }}\bra{H_{\zeta }}\right)_{|\chi }   \tensorzeta         \ \ket{G_{\overline{\zeta }}}\bra{H_{\overline{\zeta }}}\\
 & =\ \left(\ket{G_{\zeta \chi }}\bra{H_{\zeta \chi }}  \tensorzeta   \ket{G_{\overline{\zeta }}}\bra{H_{\overline{\zeta }}}\right)\braket{H_{\zeta \overline{\chi }}|G_{\zeta \overline{\chi }}}\\
 & =\ket{G_{\zeta \chi \cup \overline{\zeta }}}\bra{H_{\zeta \chi \cup \overline{\zeta }}}\braket{H_{\zeta \overline{\chi }}|G_{\zeta \overline{\chi }}}\\
\left((( .)_{\chi }  \tensorzeta   I)\left(\ket{G}\bra{H}\right)\right)_{|\emptyset } & =\braket{H_{\zeta \chi \cup \overline{\zeta }}| G_{\zeta \chi \cup \overline{\zeta }}}\braket{H_{\zeta \overline{\chi }}|G_{\zeta \overline{\chi }}}\\
\text{By Lem. [F].2} & =\braket{H|G}
\end{align*}
\end{proof}

\paragraph*{Proof of Prop. \ref{prop:gatelocality}}
\begin{proof}

[Preliminary]

\begin{align*}
A\tensorchi I & =\left(\sum _{G ,H \in \mathcal{G}}\bra{H} A\ket{G}\ket{H}\bra{G}\right) \tensorchi \left(\sum _{K \in \mathcal{G}}\ket{K}\bra{K}\right)\\
 & =\left(\sum _{G,H \in \mathcal{G}}\bra{H} A\ket{G}\ket{H}\bra{G}\right) \tensorchi \left(\sum _{G' ,H' \in \mathcal{G}}\braket{H'|G'}\ket{H'}\bra{G'}\right)\\
 & =\sum _{ \begin{array}{l}
G,H \in \mathcal{G}\\
G',H' \in \mathcal{G}
\end{array}}\bra{H} A\ket{G}\braket{H'|G'}\left(\ket{H}\bra{G} \tensorchi \ket{H'}\bra{G'}\right)\\
\text{By Lem. [F].6} & =\sum _{G,H\in \mathcal{G}}\bra{H_{\chi }} A\ket{G_{\chi }}\braket{H_{\overline{\chi }}|G_{\overline{\chi }}}\ket{H}\bra{G}
\end{align*}

[First part]

\begin{align*}
A\ \chi \text{-local} & \Leftrightarrow \bra{H} A\ket{G} =\bra{H_{\chi }} A\ket{G_{\chi }}\braket{H_{\overline{\chi }}|G_{\overline{\chi }}}\\
 & \Leftrightarrow A=\sum _{G,H\in \mathcal{G}}\bra{H_{\chi }} A\ket{G_{\chi }}\braket{H_{\overline{\chi }}|G_{\overline{\chi }}}\ket{H}\bra{G} \ =\ A\tensorchi I
\end{align*}

[Second part]

$\bra{H}( B\tensorchi I)\ket{G} =\bra{H_{\chi }} B\ket{G_{\chi }}\braket{H_{\overline{\chi }}|G_{\overline{\chi }}}$ by the preliminaries. So, $B\tensorchi I$ is $\chi $-local.

We show that $( B\tensorchi I) =(( B\tensorchi I) \tensorchi     \ I)$. By preliminaries:
\begin{align*}
\bra{H}(( B\tensorchi I) \tensorchi     \ I)\ket{G} & =\bra{H_{\chi }}( B\tensorchi I)\ket{G_{\chi }}\braket{H_{\overline{\chi }}|G_{\overline{\chi }}}\\
\text{By prelim. } & =\bra{H_{\chi \chi }} B\ket{G_{\chi \chi }} \braket{ H_{\chi \overline{\chi }} | G_{\chi \overline{\chi }}}\braket{H_{\overline{\chi }}|G_{\overline{\chi }}}\\
\text{By idempotency and Lem. [F].3} \  & =\bra{H_{\chi }} B\ket{G_{\chi }}\braket{\emptyset |\emptyset }\braket{H_{\overline{\chi }}|G_{\overline{\chi }}}\\
 & =\bra{H_{\chi }} B\ket{G_{\chi }}\braket{H_{\overline{\chi }}|G_{\overline{\chi }}}\\
\text{By prelim. } & =\bra{H}( B\tensorchi I)\ket{G}
\end{align*}
%
\end{proof}

\paragraph*{Proof of Prop. \ref{prop:duallocality}}

\begin{proof}

[$\Rightarrow $]
\begin{align*}
\left( A\braket{G|H}\right)_{|\emptyset } & =\bra{H} A\ket{G}\\
 & =\bra{H_{\chi }} A\ket{G_{\chi }}\braket{H_{\overline{\chi }}|G_{\overline{\chi }}}\\
 & =\left( A\ket{G_{\chi }}\bra{H_{\chi }}\braket{H_{\overline{\chi }}|G_{\overline{\chi }}}\right)_{|\emptyset }\\
 & =\left( A\left(\braket{G|H}\right)_{|\chi }\right)_{|\emptyset }
\end{align*}
[$\Leftarrow $] 
\begin{align*}
\bra{H} A\ket{G} & =\left( A\ket{G}\bra{H}\right)_{|\emptyset }\\
 & =\left( A\left(\ket{G}\bra{H}\right)_{|\chi }\right)_{|\emptyset }\\
 & =\left( A\ket{G_{\chi }}\bra{H_{\chi }}\braket{H_{\overline{\chi }}|G_{\overline{\chi }}}\right)_{|\emptyset }\\
 & =\bra{H_{\chi }} A\ket{G_{\chi }}\braket{H_{\overline{\chi }}|G_{\overline{\chi }}}
\end{align*}
\end{proof}

\paragraph*{Proof of Prop. \ref{prop:tomography}}

\begin{proof}

In general, \ $\rho _{|}{}_{\chi } =\sum _{G_{\chi } ,\ H_{\chi } \ \in \ \mathcal{G}_{\chi }} \alpha _{G_{\chi } H_{\chi }}\ket{G_{\chi }}\bra{H_{\chi }} \ $ and $\sigma _{|}{}_{\chi } =\sum _{G_{\chi } ,\ H_{\chi } \ \in \ \mathcal{G}_{\chi }} \beta _{G_{\chi } H_{\chi }}\ket{G_{\chi }}\bra{H_{\chi }} \ $.

Let $E_{\chi }^{H_{\chi } G_{\chi }} :=\ket{H_{\chi }}\bra{G_{\chi }}$ and $E^{H_{\chi } G_{\chi }} :=E_{\chi }^{H_{\chi } G_{\chi }} \tensorchi     \ I$, which is local by Prop. \ref{prop:gatelocality}.

We have $\left( E^{H_{\chi } G_{\chi }} \rho \right)_{|\emptyset } =\left( E^{H_{\chi } G_{\chi }} \rho _{|}{}_{\chi }\right)_{|\emptyset } =\left( E_{\chi }^{H_{\chi } G_{\chi }} \rho _{|}{}_{\chi }\right)_{|\emptyset } =\alpha _{G_{\chi } H_{\chi }}$, as the following shows:
\begin{align*}
\left( E^{H_{\chi } G_{\chi }} \rho _{|}{}_{\chi }\right)_{|\emptyset } & =\sum _{G'_{\chi } ,\ H'_{\chi } \ \in \ \mathcal{G}_{\chi } \ K\ \in \ \mathcal{G}} \alpha _{G'_{\chi } H'_{\chi }}\left(\left(\ket{H_{\chi }} \tensorchi \ket{K}\right)\left(\bra{G_{\chi }} \tensorchi \bra{K}\right)\ket{G'_{\chi }}\bra{H'_{\chi }}\right)_{|\emptyset }\\
 & =\sum _{G'_{\chi } ,\ H'_{\chi } \ \in \ \mathcal{G}_{\chi } \ K\ \in \ \mathcal{G}} \alpha _{G'_{\chi } H'_{\chi }}\left(\bra{G_{\chi }} \tensorchi \bra{K}\right)\ket{G'_{\chi }}\bra{H'_{\chi }}\left(\ket{H_{\chi }} \tensorchi \ket{K}\right)\\
\left(\bra{G_{\chi }} \tensorchi \bra{G_{\overline{\chi }}}\right)\ket{G'_{\chi }} & =\left(\bra{G_{\chi }} \tensorchi \bra{G_{\overline{\chi }}}\right)\left(\ket{G'_{\chi \chi }} \tensorchi \ket{G'_{\chi \overline{\chi }}}\right)\\
\text{By idempotency.} & =\left(\bra{G_{\chi }} \tensorchi \bra{G_{\overline{\chi }}}\right)\left(\ket{G'_{\chi }} \tensorchi \ket{\emptyset }\right)\\
 & =\braket{G_{\chi }|G'_{\chi }}\braket{G_{\overline{\chi }}|\emptyset }\\
\left( E^{H_{\chi } G_{\chi }} \rho _{|}{}_{\chi }\right)_{|\emptyset } & =\sum _{ \begin{array}{l}
G'_{\chi } ,\ H'_{\chi } \ \in \ \mathcal{G}_{\chi } \ K\ \in \ \mathcal{G}\\
\ket{G_{\chi }} \tensorchi K\ \neq \ 0\\
\ket{H_{\chi }} \tensorchi K\ \neq \ 0
\end{array}} \alpha _{G'_{\chi } H'_{\chi }}\braket{G_{\chi }|G'_{\chi }}\braket{K|\emptyset }\braket{H'_{\chi }|H_{\chi }}\braket{\emptyset |K}\\
 & =\alpha _{G_{\chi } H_{\chi }}
\end{align*}
so $\chi $-local "measurements" can tell any difference between $\rho _{|}{}_{\chi }$ and $\sigma _{|}{}_{\chi }$.

%
%
%
\end{proof}

\paragraph*{Proof of Prop. \ref{prop:dualcausality}}

\begin{proof}

[$\Rightarrow $] \ 

For all $\rho $, $\left( U\rho U^{\dagger }\right)_{|\zeta } =\left( U\rho _{|\chi } U^{\dagger }\right)_{|\zeta }$.\\ 
By Prop. \ref{prop:duallocality}, $A$ $\zeta $-local entail $\left( AU\rho U^{\dagger }\right)_{|\emptyset } =\left( AU\rho _{|\chi } U^{\dagger }\right)_{|\emptyset }$.\\
Thus, for all $\rho $, $\left( U^{\dagger } AU\rho \right)_{|\emptyset } =\left( U^{\dagger } AU\rho _{|\chi }\right)_{|\emptyset }$. 

So, by Prop. \ref{prop:duallocality}, $B=U^{\dagger } AU$ is $\chi $-local.

[ Strict$\Rightarrow $]

If $A$ is strictly $\zeta $-local then $A^{\dagger } A$ and $AA^{\dagger }$ are $\zeta $-local. \\
From the above it follows that $U^{\dagger } A^{\dagger } AU$ and $UAA^{\dagger } U$ are $\chi $-local.\\ 
But $U^{\dagger } A^{\dagger } AU=U^{\dagger } A^{\dagger } UU^{\dagger } AU=B^{\dagger } B$ and $UAA^{\dagger } U=UAUU^{\dagger } A^{\dagger } U=BB^{\dagger }$. \\
So, $B=U^{\dagger } AU$ is strictly $\chi $-local.

[$\Leftarrow $] 

For all $A$ $\zeta $-local, $U^{\dagger } AU$ is $\chi $-local. By Prop. \ref{prop:duallocality}, for all $\rho $, $\left( U^{\dagger } AU\rho \right)_{|\emptyset } =\left( U^{\dagger } AU\rho _{|\chi }\right)_{|\emptyset }$, from which it follows that for all $A$ (n.-p.), for all $\rho $, $\left( AU\rho U^{\dagger }\right)_{|\emptyset } =\left( AU\rho _{|\chi } U^{\dagger }\right)_{|\emptyset }$.\\
Finally by Prop. \ref{prop:tomography}, for all $\rho $ , we have$\left( U\rho U^{\dagger }\right)_{|\zeta } =\left( U\rho _{|\chi } U^{\dagger }\right)_{|\zeta }$. 

Thus $U$ is $\chi \zeta $-causal. 
\end{proof}

\section{Proof of the theorem}\label{app:th}

We will need two further propositions to get there, whose role is to extend a unitary to a larger space, whilst conserving its properties.

\begin{proposition}[Unitary extension]\label{prop:unitaryextension}
Consider $ \chi $ pointwise.\\
If $ U$ is a name-preserving operator over $ \mathcal{H}_{\chi }$, then $ U$ $ \chi $-consistent-preserving.\\
If $ U$ is a name-preserving unitary over $ \mathcal{H}_{\chi }$, then $ U':=U\tensorchi I$ is a name-preserving unitary with $ U^{\prime \dagger } =U^{\dagger } \tensorchi I$.
\end{proposition}
\begin{proof}

[Consistency-preservation]


 $\ket{G'_{\chi }} \tensorchi \ket{G_{\overline{\chi }}} \neq 0$ and $\ket{G_{\chi }} \tensorchi \ket{G'_{\overline{\chi }}} \neq 0$.

We need to prove first that $\bra{G'_{\chi }} U\ket{G_{\chi }} \neq 0$ entails $\ket{G'_{\chi }} \tensorchi \ket{G_{\overline{\chi }}} \neq 0$ and second that $\bra{G'_{\chi }} U^{\dagger }\ket{G_{\chi }} \neq 0$ entails $\ket{G'_{\chi }} \tensorchi \ket{G_{\overline{\chi }}} \neq 0$.

Let us prove the first.

Since $U$ is over $\mathcal{H}_{\chi }$, 
$U\ket{G_{\chi }}  =\sum _{H_{\chi } \ \in \ \mathcal{G}_{\chi }} U_{H_{\chi } G_{\chi }}\ket{H_{\chi }}$.\\
For any $G$, consider $H_{\chi }$ such that $\bra{H_{\chi }} U\ket{G_{\chi }} \neq 0$ and construct the configuration $G'=H_{\chi } \cup G_{\overline{\chi }}$.\\ 
Notice that this union is always defined since $U$ is assumed name-preserving, and hence $\mathcal{N}[ V( G_{\chi })] \cap \mathcal{N}[ V( G_{\overline{\chi }})] =\emptyset $ entails $\mathcal{N}[ V( H_{\chi })] \cap \mathcal{N}[ V( G_{\overline{\chi }})] =\emptyset $.\\
 $\neg ( V( G_{\chi }) \land V( G_{\overline{\chi }}))$ entails $\neg ( V( H_{\chi }) \land V( G_{\overline{\chi }}))$. \\
Moreover since $\chi $ is pointwise it verifies that $G'_{\chi } =H_{\chi }$ and $G'_{\overline{\chi }} =G_{\overline{\chi }}$.\\
Thus, $\ket{H_{\chi }} \tensorchi \ket{G_{\overline{\chi }}} =\ket{G'} \neq 0$.\\
It follows that $\left( U\ket{G_{\chi }}\right)$, $\ket{G_{\overline{\chi }}}$ are $\chi $-consistent.

Similarly for $U^{\dagger }$. Hence $U$ is $\chi $-consistent-preserving.

[Unitarity]

By Lem. [F].11,
\begin{align*}
( U\tensorchi I)\left( U^{\dagger } \tensorchi I\right) & =\left( UU^{\dagger } \tensorchi I\right) =( I\tensorchi I) =I\\
\left( U^{\dagger } \tensorchi I\right)( U\tensorchi I) & =\left( U^{\dagger } U\tensorchi I\right) =( I\tensorchi I) =I
\end{align*}
\begin{align*}
\left( U\ket{G_{\chi }}\right) \tensorchi \ket{G_{\overline{\mu }}} & =\sum _{H_{\chi } \ \in \ \mathcal{G}_{\chi }} U_{H_{\chi } G_{\chi }}\left(\ket{H_{\chi }} \tensorchi \ket{G_{\overline{\chi }}}\right)\\
||U'\ket{G} ||^{2} & =\sum _{H_{\chi } ,\ H'_{\chi } \in \ \mathcal{G}_{\chi }} U_{H'_{\chi } G_{\chi }}^{*} U_{H_{\chi } G_{\chi }}\left(\bra{H'_{\chi }} \tensorchi \bra{G_{\overline{\chi }}}\right)\left(\ket{H_{\chi }} \tensorchi \ket{G_{\overline{\chi }}}\right)\\
 & =\sum _{ \begin{array}{l}
H_{\chi } ,\ H'_{\chi } \in \ \mathcal{G}_{\chi }\\
\ket{H_{\chi }} \tensorchi \ket{G_{\overline{\chi }}} \ \neq \ 0\\
\ket{H'_{\chi }} \tensorchi \ket{G_{\overline{\chi }}} \ \neq \ 0
\end{array}} U_{H'_{\chi } G_{\chi }}^{*} U_{H_{\chi } G_{\chi }}\braket{H'_{\chi }|H_{\chi }}\braket{G_{\overline{\chi }}|G_{\overline{\chi }}}\\
\text{By} \ U,\ I\ \chi \text{-consistent-preserving.} & =\sum _{H_{\chi } ,\ H'_{\chi } \in \ \mathcal{G}_{\chi }} U_{H'_{\chi } G_{\chi }}^{*} U_{H_{\chi } G_{\chi }}\braket{H'_{\chi }|H_{\chi }}\braket{G_{\overline{\chi }}|G_{\overline{\chi }}}\\
\text{By unitarity of} \ U. & =\sum _{H_{\chi } \ \in \ \mathcal{G}_{\chi }} U_{H_{\chi } G_{\chi }}^{*} U_{H_{\chi } G_{\chi }} =1
\end{align*}
Thus $U'$ is an isometry, i.e. $U^{\prime \dagger } U'=I$. 

Notice that $\left( U^{\dagger } \tensorchi I\right)$ is its right inverse since by means of Lem. [F].11. 

Since $U$ is over $\mathcal{H}_{\chi }$, we have that $U^{\ \dagger }$preserves the range of $\chi $, and so $U^{\dagger } ,\ I$ are $\chi $-consistent-preserving. 

Thus, 
\begin{align*}
\left( U^{\dagger } \tensorchi I\right)\ket{G} & =\sum _{H_{\chi } \ \in \ \mathcal{G}_{\chi }} U_{G_{\chi } H_{\chi }}^{*}\left(\ket{H_{\chi }} \tensorchi \ket{G_{\overline{\chi }}}\right)\\
U'\left( U^{\dagger } \tensorchi I\right)\ket{G} & =\sum _{ \begin{array}{l}
H_{\chi } \ \in \ \mathcal{G}_{\chi }\\
\ket{H_{\chi }} \tensorchi \ket{G_{\overline{\chi }}} \ \neq \ 0
\end{array}} U_{H'_{\chi } H_{\chi }} U_{G_{\chi } H_{\chi }}^{*}\left(\ket{H'_{\chi }} \tensorchi \ket{G_{\overline{\chi }}}\right)\\
\text{By} \ U^{\dagger } ,\ I\ \chi \text{-consistent-preserving.} & =\sum _{H_{\chi } \ \in \ \mathcal{G}_{\chi }} U_{H_{\chi } H'_{\chi }} U_{G_{\chi } H_{\chi }}^{*}\left(\ket{H'_{\chi }} \tensorchi \ket{G_{\overline{\chi }}}\right)\\
 & =\sum _{H_{\chi } \ \in \ \mathcal{G}_{\chi }} I_{H'_{\chi } G_{\chi }}\left(\ket{H'_{\chi }} \tensorchi \ket{G_{\overline{\chi }}}\right)\\
 & =\ket{G}
\end{align*}
[Name-preservation]

Follows from Prop. \ref{prop:gatelocality}.

\end{proof}

\begin{proposition}[Causal extension]\label{prop:causalextension}
First consider $ U$ a $ \chi '\zeta '$-causal operator and $ \chi '\sqsubseteq \chi $, $ \zeta \sqsubseteq \zeta '$. \\
Then $ U$ is an $ \chi \zeta $-causal operator.

Second consider three restrictions $ \mu ,\zeta, \chi $ such that $ [ \mu ,\zeta ] =[\overline{\mu } ,\zeta ] =[ \mu ,\overline{\zeta }] =[\overline{\mu } ,\overline{\zeta }] =0$ with $ \mu $ pointwise and both $\zeta$ and $\chi$ extensible. \\
Consider $ U$ an name-preserving $ \chi \zeta $-causal unitary operator over $ \mathcal{H}_{\mu }$, and $ U':=U \tensormu  I$ its unitary extension. \\
Then $ U'$ is $ \xi \zeta $-causal operator, w.r.t the restriction $ \xi :=\mu \chi \cup \overline{\mu } \zeta $.
\end{proposition}

\begin{proof}

[First part]

Suppose $U$ is $\chi '\zeta '$-causal and $\chi '\sqsubseteq \chi $, $\zeta \sqsubseteq \zeta '$. 

Prop. \ref{prop:dualcausality} and using Lem. [F].8, $A$ $\zeta $-local implies $A$ $\zeta '$-local implies \ $U^{\dagger } AU$ $\chi '$-local implies $U^{\dagger } AU\ \chi $-local. Thus $A$ $\zeta $-local implies $U^{\dagger } AU$ $\chi $-local which by Prop. \ref{prop:dualcausality} is equivalent to $\chi \zeta $-causality.

[Second part]

From Lem. [F].5 we have $[ \mu ,\xi ] =[\overline{\mu } ,\xi ] =[ \mu ,\overline{\xi }] =[\overline{\mu } ,\overline{\xi }] =0$, and $\xi $ a restriction.

Any $A$ $\zeta $-local is of the form $A=L \tensorzeta   I$ with $L=\sum \alpha _{G_{\zeta } H_{\zeta }}\ket{G_{\zeta }}\bra{H_{\zeta }}$.
\begin{align*}
\ket{G_{\zeta }}\bra{H_{\zeta }}  \tensorzeta   I & =\ \left(\ket{G_{\zeta \mu }}\bra{H_{\zeta \mu }}  \tensormu  \ket{G_{\zeta \overline{\mu }}}\bra{H_{\zeta \overline{\mu }}}\right)  \tensorzeta   ( I_{\overline{\zeta } \mu }  \tensormu  I_{\overline{\zeta }\overline{\mu }})\\
\text{Commut. \& Lem. [F].7.} & =\ \left(\ket{G_{\mu \zeta }}\bra{H_{\mu \zeta }}  \tensorzeta   I_{\mu \overline{\zeta }}\right) \tensormu      \left(\ket{G_{\overline{\mu } \zeta }}\bra{H_{\overline{\mu } \zeta }}  \tensorzeta   I_{\overline{\mu }\overline{\zeta }}\right)\\
U'\left(\ket{G_{\zeta }}\bra{H_{\zeta }}  \tensorzeta   I\right) U^{\prime \dagger } & =\left( U\left(\ket{G_{\mu \zeta }}\bra{H_{\mu \zeta }}  \tensorzeta   I_{\mu \overline{\zeta }}\right) U^{\dagger }\right)  \tensormu  \left(\ket{G_{\overline{\mu } \zeta }}\bra{H_{\overline{\mu } \zeta }}  \tensorzeta   I_{\overline{\mu }\overline{\zeta }}\right)\\
\text{(By dual caus.) } & =\left( M^{G_{\mu \zeta } H_{\mu \zeta }} \tensorchi I_{\mu \chi }\right) \tensormu      \left(\ket{G_{\overline{\mu } \zeta }}\bra{H_{\overline{\mu } \zeta }}  \tensorzeta   I_{\overline{\mu }\overline{\zeta }}\right)\\
U\ \text{over} \ \mathcal{H}_{\mu } & =\left( M^{G_{\mu \zeta } H_{\mu \zeta }} \ \tensormuchi \ I_{\mu \chi }\right) \tensormu      \left(\ket{G_{\overline{\mu } \zeta }}\bra{H_{\overline{\mu } \zeta }} \ \tensormubarchi \ I_{\overline{\mu }\overline{\zeta }}\right)\\
 & =\left( M^{G_{\mu \zeta } H_{\mu \zeta }}  \tensorxi    I_{\mu \overline{\xi }}\right) \tensormu      \left(\ket{G_{\overline{\mu } \zeta }}\bra{H_{\overline{\mu } \zeta }}  \tensorxi    I_{\overline{\mu }\overline{\xi }}\right)\\
\text{Commut. \& Lem. [F].7.} & =\left( M^{G_{\mu \zeta } H_{\mu \zeta }}  \tensormu  \ket{G_{\overline{\mu } \zeta }}\bra{H_{\overline{\mu } \zeta }}\right)  \tensorxi    ( I_{\mu }  \tensormu  I_{\overline{\mu }})\\
\text{Lem. [F].6.} & =\left( M^{G_{\mu \zeta } H_{\mu \zeta }}  \tensormu  \ket{G_{\overline{\mu } \zeta }}\bra{H_{\overline{\mu } \zeta }}\right)  \tensorxi    I_{\overline{\xi }}\\
U'A U^{\prime \dagger } & =\sum \alpha _{G_{\zeta } H_{\zeta }}\left(\left( M^{G_{\mu \zeta } H_{\mu \zeta }}  \tensormu  \ket{G_{\overline{\mu } \zeta }}\bra{H_{\overline{\mu } \zeta }}\right)  \tensorxi    I_{\overline{\xi }}\right)\\
\text{By bilinearity } & =\left(\sum \alpha _{G_{\zeta } H_{\zeta }}\left( M^{G_{\mu \zeta } H_{\mu \zeta }}  \tensormu  \ket{G_{\overline{\mu } \zeta }}\bra{H_{\overline{\mu } \zeta }}\right)\right)   \tensorxi          I_{\overline{\xi }}
\end{align*}
So, $U'A U^{\prime \dagger }$ is $\xi $-local by the Prop. \ref{prop:gatelocality}. By Prop. \ref{prop:unitaryextension}, $U'$ is name-preserving. By Prop. \ref{prop:dualcausality}, \ $U'$ is $\xi \zeta $-causal. 
\end{proof}

\paragraph*{Proof of Th. \ref{th:blockdecomposition}}

\begin{proof}
Clearly $[ \mu ,\zeta ] =[\overline{\mu } ,\zeta ] =[ \mu ,\overline{\zeta }] =[\overline{\mu } ,\overline{\zeta }] =0$ as both are pointwise.

By Prop. \ref{prop:unitaryextension}, $U':=U \tensormu  I$ is unitary over $\mathcal{H} '$.\\

By Prop. \ref{prop:causalextension} and since $U$ is $\chi _{v} \zeta '_{v}$-causal, it is $\chi _{v} \zeta _{v}$-causal, and $U'$ is $\xi _{v} \zeta _{v}$-causal, with $\xi _{v} :=\mu \chi _{v} \cup \overline{\mu } \zeta _{v}$ a restriction.

Let the toggle $\tau $ be the bijection over systems such that $\tau ( b.\sigma .u) =\neg b.\sigma .u$.\\
Extend $\tau $ to $\mathcal{G} '$ by acting pointwise upon each system, and to $\mathcal{H} '$ by linearity.\\
Notice that it is unitary and name-preserving.\\
Notice that $\tau \left(\ket{G_{\mu }}  \tensormu  \ket{G_{\overline{\mu }}}\right) =\left( \tau \ket{G_{\overline{\mu }}}  \tensormu  \tau \ket{G_{\mu }}\right)$.\\
It is also unitary over $\mathcal{H}_{\zeta _{v}}$, thus $\tau _{v} :=\tau \   \tensorzeta         _{v} \ I$ is unitary over $\mathcal{H} '$ by Prop. \ref{prop:unitaryextension}.\\
By Prop. \ref{prop:gatelocality}, $\tau _{v}$ is $\zeta _{v}$-local. By unitarity, it is strictly $\zeta _{v}$-local.\\
Moreover, 
\begin{equation*}
\left(\prod _{v\ \in \ \mathcal{V}} \tau_{v}\right)=\tau
\end{equation*}
Notice also that $[ \tau _{u} ,\tau _{v}] =0$.

Let $K_{v} :=U^{\prime \dagger } \tau _{v} U'$. \\
It is name-preserving as a composition of name-preserving operators.\\
Since adjunction by a unitary is a morphism, $[ K_{u} ,K_{v}] =0$.\\
By Prop. \ref{prop:dualcausality}, it is $\xi _{v}$-local. By unitarity, it is strictly $\xi _{v}$-local.

Finally,
\begin{align*}
\left(\prod _{v\ \in \ \mathcal{V}}\tau _{v}\right)\left(\prod _{v\ \in \ \mathcal{V}}K_{v}\right)\ket{G} & =\tau \dotsc \left(U'^{\dagger }\tau _{v_{2}}U'\right)\left(U'^{\dagger }\tau _{v_{1}}U'\right)\ket{G}\\
\text{By unitarity of }U'. & =\tau \ U'^{\dagger }\ \left(\prod _{v\ \in \ \mathcal{V}}\tau _{v}\right)\ U'\ \ket{G}\\
 & =\tau \ U^{\prime \dagger } \ \tau \ U'\ \ket{G}\\
\text{By Prop. \ref{prop:unitaryextension}} & =\tau \ \left( U^{\dagger }  \tensormu  I\right) \ \tau \ ( U \tensormu  I) \ \left(\ket{G_{\mu }}  \tensormu  \ket{G_{\overline{\mu }}}\right) \ \\
 & =\tau \ \left( U^{\dagger }  \tensormu  I\right) \ \tau \ \left( U\ket{G_{\mu }}  \tensormu  \ket{G_{\overline{\mu }}}\right)\\
\text{Since} \ U\ \text{preserves the range of } \mu . & =\tau \ \left( U^{\dagger }  \tensormu  I\right)\left( \tau \ket{G_{\overline{\mu }}}  \tensormu  \tau U\ket{G_{\mu }}\right)\\
 & =\tau \ \left( U^{\dagger } \tau \ket{G_{\overline{\mu }}}  \tensormu  \tau U\ket{G_{\mu }}\right)\\
\text{Since} \ U\ \text{preserves the range of } \mu . & =\tau ^{2} U\ket{G_{\mu }}  \tensormu  \tau \ U^{\dagger } \tau \ket{G_{\overline{\mu }}}\\
\text{Since} \ \tau \ \text{involutive.} & =U\ket{G_{\mu }}  \tensormu  \tau \ U^{\dagger } \tau \ket{G_{\overline{\mu }}}\\
\left(\prod _{v\ \in \ \mathcal{V}}\tau_{v}\right)\left(\prod _{v\ \in \ \mathcal{V}}K_{v}\right)\ket{G_{\mu }} & =\left(\prod _{v\ \in \ \mathcal{V}}\tau _{v}\right)\left(\prod _{v\ \in \ \mathcal{V}}K_{v}\right)\left(\ket{G_{\mu }} \tensormu  \ket{\emptyset }\right)\\
 & =U\ket{G_{\mu }}  \tensormu  \tau \ U^{\dagger } \tau \ket{\emptyset }\\
\text{By n.-p.}, & =U\ket{G_{\mu }}  \tensormu  \ket{\emptyset }\\
 & =U\ket{G_{\mu }}
\end{align*}
\end{proof}

\section{Table of Results}
\begin{table}
\caption{\label{tab:toolbox} {\em Mathematical toolbox}.}  
\renewcommand{\arraystretch}{2}        
\begin{tabular}{p{0.22\textwidth}p{0.78\textwidth}}
\hline 
  Lem. [F].2 & $ \braket{H|G} =\braket{H_{\chi }|G_{\chi }}\braket{H_{\overline{\chi }}|G_{\overline{\chi }}}$  \\
\hline 
  Lem. [F].3 & $ \chi \chi =\chi $ \qquad $ \chi \overline{\chi } =\emptyset $  \\
\hline 
 ~ & $ \left( \rho \ket{G}\bra{H}\right)_{|\emptyset } =\bra{H} \rho \ket{G}$ \qquad $ ( \rho A)_{|\emptyset } =( A\rho )_{|\emptyset }$ \qquad $ ( \alpha \rho )_{|\chi } =\alpha ( \rho )_{|\chi }$  \\
\hline 
~ & If $ \ket{G} \tensorchi \ket{G'} \neq 0$\\  & then $ \ket{G} \tensorchi \ket{G'} \ =\ \ket{G\cup G'}$ \quad and \quad $ \left(\ket{G} \tensorchi \ket{G'}\right)_{\chi } =\ket{G}$  \\
\hline 
~ & $ \ \rho _{|\chi } \ =\ \sum _{G,H\in \mathcal{G} ,\ G_{\overline{\chi }} =H_{\overline{\chi }}} \rho _{G H} \ \ket{G_{\chi }}\bra{H_{\chi }} \ $  \\
\hline 
Lem. [F].8 & If $ \zeta \sqsubseteq \chi $, $ \begin{aligned}
( \rho _{|\chi })_{|\zeta } & =\rho _{|\zeta }
\end{aligned}$ and \quad $A$ $ \zeta $-local is $ \chi $-local. \\
\hline 
 Lem. [F].6 & $ \ A\tensorchi B\ =\ \sum _{G,H\in \mathcal{G}} A_{G_{\chi } H_{\chi }} B_{G_{\overline{\chi }} H_{\overline{\chi }}} \ \ket{G}\bra{H} \ $ \qquad $ A\tensorchi I =A\tensorchi I_{\overline{\chi }}$  \\
\hline 

 Lem. [F].7 & If $ [ \chi ,\zeta ] =[\overline{\chi } ,\zeta ] =[ \chi ,\overline{\zeta }] =[\overline{\chi } ,\overline{\zeta }] =0$,\\ & then $ ( A \tensorzeta   B) \tensorchi ( C \tensorzeta   D) =( A\tensorchi C)  \tensorzeta   ( B\tensorchi D) \ $  \\
\hline 
 Lem. [F].9 & If $ \rho ,\sigma $ $ \chi $-consistent, \ $ ( \rho \tensorchi \sigma )_{|\chi } =\rho \ \sigma _{|\emptyset } \ $\\ & If $ \zeta \sqsubseteq \chi $, $ \rho ,\sigma $ $ \chi $-consistent, \ $ ( \rho \tensorchi \sigma )_{|\zeta } =\rho _{|\zeta } \ \sigma _{|\emptyset } \ $  \\
\hline 
 Lem. [F].10 & If $ [ \chi ,\zeta ] =[\overline{\chi } ,\zeta ] =[ \chi ,\overline{\zeta }] =[\overline{\chi } ,\overline{\zeta }] =0$, \ and $ \rho ,\sigma $ $ \chi $-consistent,\\ & $ ( \rho \tensorchi \sigma )_{|\zeta } =\rho _{|\zeta } \tensorchi \sigma _{|\zeta } \ $  \\
\hline 
 Lem. [F].11 &  $ ( A\tensorchi I)\ket{G} =A\ket{G_{\chi }} \tensorchi \ket{G_{\overline{\chi }}}$.\\ & If $ A$, $ A'$, $ B$, $ B'$ are $ \chi $-consistent-preserving,$ ( A'\tensorchi B')( A\tensorchi B) =\ A'A\tensorchi B'B$.  \\
 \hline
\end{tabular}
\end{table}

\end{document}